\documentclass[twocolumn,10pt,notitlepage,nofootinbib,superscriptaddress,numerical]{revtex4-2}

\usepackage[utf8]{inputenc}
\usepackage[colorlinks=true,linkcolor=blue,citecolor=blue,urlcolor=blue]{hyperref}

\usepackage{amsmath,amssymb}
\usepackage{dsfont}
\usepackage{xcolor}
\usepackage{amsthm}

\usepackage{algorithm}
\usepackage{algpseudocode}

\algnewcommand\algorithmicforeach{\textbf{for each}}
\algdef{S}[FOR]{ForEach}[1]{\algorithmicforeach\ #1\ \algorithmicdo}

\usepackage{ragged2e}

\newtheorem{theorem}{Theorem}[section]

\newtheorem{prop}[theorem]{Proposition}
\newtheorem{lemma}[theorem]{Lemma}

\theoremstyle{definition}
\newtheorem{definition}[theorem]{Definition}
\newtheorem{example}{Example}[section]
\newtheorem{remark}{Remark}

\definecolor{FGreen}{RGB}{1,100,10}

\usepackage{graphicx}
\usepackage[caption=false]{subfig}
\graphicspath{
    {figures/}
}

\usepackage{mathtools} 
\usepackage[shortlabels]{enumitem}
\usepackage{bm}
\usepackage{dsfont} 
\usepackage{etoolbox}

\docsvlist{A,B,C,D,E,F,G,H,I,J,K,L,M,N,O,P,Q,R,S,T,U,V,W,X,Y,Z}

\docsvlist{A,B,C,D,E,F,G,H,I,J,K,L,M,N,O,P,Q,R,S,T,U,V,W,X,Y,Z}

\docsvlist{A,B,C,D,E,F,G,H,I,J,K,L,M,N,O,P,Q,R,S,T,U,V,W,X,Y,Z}

\docsvlist{d}

\docsvlist{E,P,R,e,k,p,q,r,g,x,y}

\usepackage{ stmaryrd }
\def\llb{\llbracket}
\def\rrb{\rrbracket}
\def\bra{\langle}
\def\ket{\rangle}
\def\lr{\leftrightarrow}
\def\backslash{\symbol{92}}

\def\weight{\mathfrak{w}}
\def\qubit{\mathfrak{q}}

\graphicspath{
    {figures/}
}

\DeclareMathOperator{\im}{im}

\DeclareMathOperator{\spn}{span}

\DeclareMathOperator{\poly}{poly}

\newcommand{\norm}[1]{\left\lVert#1\right\rVert}

\usepackage{tikz-cd} 
\usetikzlibrary{matrix,arrows,decorations.pathmorphing}
\usetikzlibrary{fit, shapes.geometric,calc}

\usepackage{xcolor}
\newcommand{\red}[1]{{\color{red} #1}}
\newcommand{\gray}[1]{{\color{gray} #1}}
\newcommand{\blue}[1]{{\color{blue} #1}}

\makeatletter
\def\l@subsection#1#2{}
\def\l@subsubsection#1#2{}
\makeatother

\begin{document}

\title{
Quantum Weight Reduction with Layer Codes
}
\author{Andrew C. Yuan}
\thanks{The first two authors contributed equally.}
\affiliation{Iceberg Quantum, Sydney}
\affiliation{Condensed Matter Theory Center and Joint Quantum Institute,
Department of Physics, University of Maryland, College Park, Maryland 20742, USA}
\author{Nouédyn Baspin}
\thanks{The first two authors contributed equally.}
\affiliation{Iceberg Quantum, Sydney}
\author{Dominic J. Williamson}
\affiliation{School of Physics, The University of Sydney, Sydney NSW 2006, Australia}

\date{March 2026}

\begin{abstract}
    %
    Quantum weight reduction procedures ease the implementation of quantum codes by sparsifying them, resulting in low-weight checks and low-degree qubits. 
    However, to date, only few quantum weight reduction methods have been explored. 
    In this work we introduce a simple and general procedure for quantum weight reduction that achieves check weight 6 and total qubit degree 6, lower than existing procedures at the cost of a potentially larger qubit overhead. 
    Our quantum weight reduction procedure replaces each qubit and check in an arbitrary Calderbank-Shor-Steane code with an ample patch of surface code, these patches are then joined together to form a geometrically nonlocal Layer Code.
    This is a quantum analog of the simple classical weight reduction procedure where each bit and check is replaced by a repetition code. 
    Due to the simplicity of our weight reduction procedure, bounds on the weight and degree of the resulting code follow directly from the Layer Code construction and hence are easily verified by inspection. 
    Our procedure is well suited for implementation in modular architectures that consist of surface code patches networked via long-range interconnects. 
\end{abstract}

\maketitle


\section{Introduction}

Quantum error-correcting codes are essential to suppress errors when storing and processing quantum information at scale. 
Recent progress has led to quantum codes that efficiently store large amounts of quantum information with a large degree of protection from noise. 
However, challenges remain that hinder the practical implementation of these codes. 
One such challenge is the search for efficient quantum codes with low-weight check operators that have a low degree of overlap on each qubit. 

A quantum \textit{low-density parity check} (LDPC) code is a quantum stabilizer code where each generating check acts on at most $\weight=O(1)$ qubits and each qubit participates in at most $\qubit=O(1)$ checks~\cite{gottesman2024surviving}. 
Beyond satisfying a notion of locality inspired by physical intuition, LDPC codes are essential for ensuring fault-tolerant quantum error correction in asymptotically large codes \cite{gottesman2013fault}.
Specifically, their local structure guarantees that if the error rate is a sufficiently small constant, so that error clusters do not percolate, then the disconnected error clusters can be corrected. 
For this reason, the search for LDPC codes with optimal scaling properties has attracted substantial attention from the research community, culminating in the breakthrough discovery of asymptotically good LDPC code whose encoding rates and relative distances are
constant \cite{panteleev2022asymptotically}.
Since then, a number of alternative constructions have been proposed \cite{leverrier2022quantum,dinur2023good}.

Good asymptotic parameters are not enough for LDPC codes to be interesting and useful beyond the realm of pure theory. 
For example, the check weight $\weight$ and qubit degree $\qubit$ may be prohibitively large constants, making implementation on realistic hardware impractical. 
Such codes stand to benefit from weight and degree reduction even if it comes at the cost of a modest qubit overhead.

Quantum weight reduction was pioneered by Hastings in Refs.~\cite{hastings2016weight,hastings2021quantum}, where he introduced a procedure that transforms any Calderbank-Shor-Steane (CSS) code with weight $\weight,\qubit$ into a constant weight CSS code with ancilla qubit overhead $O(\poly (\weight,\qubit))$.
While this initial approach reduced the code distance by a factor of $\poly(\weight,\qubit)$, subsequent work demonstrated that the distance can be preserved up to an $\Omega(1)$ constant by incorporating expander graphs into the construction -- an idea that first appeared in the context of low-overhead quantum code surgery \cite{williamson2024low,ide2025fault}.
Next, Ref.~\cite{sabo2024weight} introduced a simplified weight reduction procedure based on classical weight reduction which only applies to hypergraph product codes~\cite{tillich2013quantum}.
Another work~\cite{baspin2024wire} applied similar techniques to construct a quantum weight reduction procedure that produces a family of subsystem codes known as Wire Codes. 
Most recently, a weight reduction method was discovered, in which the overhead was substantially lowered to $O( \weight\qubit \log (\weight \qubit))$, while preserving the code distance up to a universal constant~\cite{hsieh2025simplified}.


Quantum weight reduction techniques share a fundamental filiation with quantum code surgery \cite{horsman2012surface,cohen2022low,cross2024improved,williamson2024low,ide2025fault,swaroop2026universal}. 
Specifically, by treating high-weight logical operators as large stabilizer checks, weight reduction schemes allow for logical information to be extracted indirectly via modified low-weight stabilizers; and vice versa high-weight stabilizers can be viewed as logical operators that are measured via code surgery.
The Layer Codes were constructed in this way, from a stack of surface code patches in three dimensions via the local measurement of concatenated stabilizers from a good LDPC code~\cite{williamson2023layer}. 
In this work we apply a similar approach to find a simple procedure for quantum weight reduction that locally implements a concatenated code on patches of surface code while maintaining low check weight and qubit degree. 

\subsection{Statement of the main result}

Our main result is a simple and explicit quantum weight reduction procedure that produces single digit reduced weights, and degrees, without requiring the use of complicated constructions such as expander graphs. This comes at the cost of an ancilla qubit overhead that has a faster asymptotic growth than the result in Ref.~\cite{hsieh2025simplified}. 
Specifically, for any input CSS code with max check weight $\weight$ and qubit degree $\qubit$, an explicit output CSS code can be constructed so that the max check weight is $6$, total qubit degree is $6$, qubit overhead is $O(\weight^4 \qubit^4)$, the logical subspace is preserved, and the distance is increased by a multiplicative factor~$\Omega(\weight\qubit^2)$. 
\begin{theorem}[Main result, corollary of Theorem \ref{thm:main}]
    \label{thm:main-intro}
    Let $D$ be an $\llb n,k,d\rrb$ CSS code with maximum weight $\weight$, and qubit degree $\qubit$; then Algorithm \ref{alg:sparsification} outputs a CSS code $D^\mathrm{sparse}$ with parameters $\llb O(\weight^4 \qubit^4 n), k, \Omega(\weight\qubit^2 d)\rrb$, maximum weight $6$, and total qubit degree $6$.
\end{theorem}
This construction generalizes the Layer Codes \cite{williamson2023layer}, which were used to embed any LDPC CSS code into three-dimensional Euclidean space.
The maximum weight and total qubit degree of the Layer Codes is 6, see Eq.~\eqref{eq:weight-diagram-2D-layer} or the explicit checks in Ref.~~\cite{williamson2023layer}. 
This is because there are three types of layers, corresponding to input data qubits, $X$-checks, and $Z$-checks, each of which hosts a modified surface code with checks (qubits) that interact with at most one additional qubit (check) in each distinct layer type. 
Importantly, in the context of weight reduction, we are not restricted by Euclidean locality.
Hence, the ancilla overhead of conventional Layer Codes has been substantially reduced here, while the reduced weights and total qubit degrees remain the same.


We now compare our construction to previous results in the literature. 
Despite the advances outlined above, significant caveats remain for the general CSS stabilizer code weight reduction procedures in Refs.~\cite{hastings2021quantum,hsieh2025simplified} that are addressed by our work.
First, these works were technically obscure, leading to in-depth follow up works clarifying the details involved in the original constructions  \cite{wills2024tradeoff,sabo2024weight,tan2025effective,parsimonious}. 
Second, due to the intricate, and at times implicit, nature of the previous constructions, errors and omissions in Refs.~\cite{hastings2021quantum,hsieh2025simplified} result in guarantees that are weaker than the advertised results. 
In particular Lemma 8 in Ref.~\cite{hastings2021quantum}  contains a minor error\footnote{In the proof, it is claimed that $X$-stabilizers of $\mathcal{A}$ have weight bounded by $\weight_X$ in the reduced cone code. However, this is not necessarily true, since there are now additional qubits coming from the chain map in Definition 2 of \cite{hastings2021quantum}.} at point 5, which implies that the reduced weights are only upper bounded by $(\weight_X,\weight_Z,\qubit_X,\qubit_Z)\le(42,36,4,3)$ instead of the claimed $(5,5,3,5)$. See Appendix \eqref{app:Hastings} for a detailed exposition.
Similarly, a careful analysis of the proof in Ref.~\cite{hsieh2025simplified} highlights some areas that require further clarification. For example, see Remark (17) in Ref. \cite{parsimonious}. 
For the purpose of asymptotic scaling bounds, these points are minor and are likely due to the technical difficulty of the associated constructions. However, the weaker guarantees they lead to may be relevant for practical applications of these results. 
Third, Refs.~\cite{hastings2021quantum,hsieh2025simplified} require a  corresponding expander graph to preserve the code distance up to $\Omega(1)$ constant.
It is well known that expander graphs are both difficult to explicitly construct and verify, hence this introduces a major obstruction to practical applications.

Our results establish a new approach to quantum weight reduction that is based on patches of surface code and simple topological defects that are described by explicit stabilizer checks~\cite{williamson2023layer}. 
We expect that fine-tuning the choice of dimensions of each surface code layer and the partitioning of these layers into disjoint sets will produce lower overhead weight reductions for specific code instances than is guaranteed by the asymptotic result in Theorem.~\ref{thm:main-intro}.

\subsection{Outlook}

We anticipate that our simplified quantum weight reduction procedure will find broad applications in constructing explicit weight-reduced codes. 
Our construction is a natural quantum extension of the simple classical weight reduction procedure that replaces each bit and check with a repetition code. 
This construction has found broad applicability due to its simplicity and ease of use, which also hold for our quantum weight reduction procedure~\cite{sabo2024weight,baspin2024wire,baspin2023combinatorial}. 
The codes output by our procedure are naturally suited to implementation in an architecture that consists of modular patches of surface codes that are networked together via sparse long-range interconnects. 
Our weight reduced codes come with a decoder as described in Ref.~\cite{gu2025layer,williamson2025partial}.
Similar arguments to Ref.~\cite{gu2025layer,williamson2025partial} would imply that weight reduced good Quantum Tanner codes will achieve self-correction under this decoder due to the preservation of the energy barrier scaling.

The Layer Code quantum weight reduction procedure raises a number of questions for future work. 
First, is it possible to reduce the overhead while maintaining the simple structure of the weight reduction procedure based on surface codes? 
Is it possible to reduce the weight and total degree to 5, at the cost of additional ancilla qubits, by decomposing each of the topological defects that appears in the construction? 
Can we expand the set of target spaces for the Layer Codes to higher-dimensional Euclidean space and beyond? In Refs.~\cite{baspin2023combinatorial, baspin2024wire}, for example, the authors obtained explicit embeddings for any Euclidean dimension, and showed that, in fact, expansion alone is sufficient to guarantee the existence of an embedding into a general graph.
In another direction, can our weight reduction procedure be extended beyond CSS codes? 
Is it possible to engineer the required long range interconnects to implement our weight-reduced codes on superconducting, neutral atom, or trapped ion hardware?
Can our codes be used as a memory in a surface code based fault-tolerant architecture to improve the performance beyond the techniques in Ref.~\cite{gidney2025yoked}?

\subsection{Section Outline}

The remaining sections are laid out as follows. 
In Section~\ref{sec:Overview}, assuming only basic knowledge of stabilizer and Pauli subgroups, we provide an overview of our construction using the seminal $\llb 9,1,3\rrb$ Shor code as an example.
In Section~\ref{sec:prelim}, we provide the necessary preliminaries in algebraic homology to understand the technical construction of our sparsification method with layer codes.
In Section~\ref{sec:algebra}, we provide the details of the construction, culminating in Theorem \ref{thm:logical} and \ref{thm:distance}.
In Section \ref{sec:proof}, we prove Theorem \ref{thm:logical} and \ref{thm:distance}.
In Appendix~\ref{app:TP} we prove a technical condition that the tensor product preserves the isoperimetric inequality (relative expansion) used in our construction. 
In Appendix~\ref{app:Hastings} we review the details of Hastings's weight reduction procedure~\cite{hastings2021quantum}.

\section{Overview of the Construction}
\label{sec:Overview}

In this section we provide an illustrated high-level overview of our quantum weight reduction construction. 
We focus on the $\llb 9,1,3\rrb$ Shor code as an intuition-building example, while the full algorithm is provided in Algorithm~\ref{alg:sparsification}.
For simplicity of presentation, we defer the explanation of some details to Section~\ref{sec:algebra} where the full algebraic construction is described along with the main results on the logical subspace and overhead (Theorem~\ref{thm:logical}) and code distance (Theorem~\ref{thm:distance}) of the weight reduced output Layer Code.

The original Layer Code construction \cite{williamson2023layer} takes as input a CSS stabilizer code on $n$ qubits. Each of these $n$ qubits are mapped to `data layers', while each of its stabilizers are mapped to `check layers'. 
Due to the connectivity constraints of the three-dimensional Euclidean space, embedding these layers requires a significant amount of space and ancillary qubits, which eventually results in a new code whose parameters are diminished with respect to the input code. 
When the connectivity constraints are removed, so is the need for the ancillary overhead. 
In this way, we can obtain a geometry-less Layer Code that sparsifies the original code, while better preserving its parameters. 

\begin{figure}[t]
    \centering
    \subfloat[\label{fig:shor} Shor Code.]{%
    \centering
    \includegraphics[scale = .8, page=1]{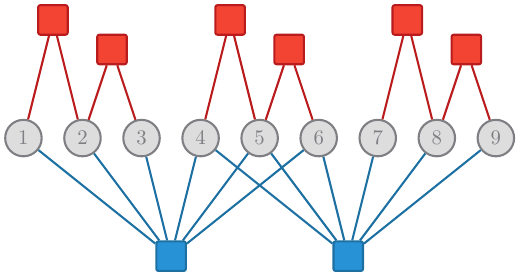}
    }
    \\
    \subfloat[\label{fig:shor-layers} Layer-concatenated Shor Code.]{%
    \centering
    \includegraphics[scale = .8, page=2]{diagrams-mod.pdf}
    }
    \caption{ (a) The grey circles denote qubits. The blue boxes denote $X$-checks, with lines denoting which qubits the check acts on, and similarly for red boxes, which denote $Z$-checks.
    (b) $X,Z$-checks and qubits are replaced with surface codes. Here, solid lines depict smooth boundaries, and dashed lines depict rough boundaries. }
\end{figure}

Loosening the connectivity constraints and removing the ambient three-dimensional space from the Layer Code construction raises new challenges. 
The key idea to compress the size of a Layer Code, is to partition its layers into sets containing layers that have no collisions. 
This allows us to use patches of surface code with sizes specified by the number of such collisionless sets, which are much smaller than the patches of surface code in the original three-dimensional Layer Code. 

As an example, the input code $A$ is chosen to be the $\llb 9,1,3\rrb$ Shor code as depicted in Fig. \ref{fig:shor}.
Similar to the conventional Layer Code construction, we start by replacing each $X$-check, $Z$-check and qubit with a surface code layer, whose detailed dimensions are determined below. 
The $X$-check layers have smooth boundary conditions, the $Z$-check layers have rough boundaries, and the qubit layers have standard planar code boundaries, following the extended Layer Code construction in Ref.~\cite{williamson2025partial}.
This is depicted in Fig. \ref{fig:shor-layers}.
The adjacency relations (red and blue lines) are then replaced by local interactions between corresponding check layers and qubits layers.
Specifically, each $Z$-check layer acts along the string-like $Z$-type logical of the corresponding qubit layers, as depicted in red in Fig. \ref{fig:shor-defect}.
To achieve a weight reduction of the output code, distinct qubit layers must interact with distinct cross sections of the $Z$-check layer.
We then repeat the same procedure for the $X$-check layers as depicted in blue in Fig. \ref{fig:shor-defect}.
However, since the $X$- and $Z$-type logical operators have nontrivial overlap on each qubit layer, a green line defect is introduced following the procedure in the conventional Layer Code construction to guarantee that the output code remains a stabilizer code, i.e., all checks commute.
The final output code $C$ for the Shor code $A$ is depicted schematically in Fig. \ref{fig:shor-dimensions}.

Our construction crucially needs to ensure that regions of distinct $X$- and $Z$-layers that interact with the same green line defect can be identified. In particular they must have the same length, which imposes rigid requirements on how the defects are arranged.
In the case of the conventional Layer Codes, the layers are first immersed in three dimensions, and their intersections naturally yield the locations for these defects. 
In the present construction, however, due to the lack of an ambient reference space, we require an alternative coordinate system to navigate the surface code patches.
For this purpose, we introduce the axes $\hat{x},\hat{q},\hat{z}$ as shown in Fig. \ref{fig:shor-dimensions}.
Specifically, each  $X$-check layer is indexed by $\hat{q}\wedge \hat{z}$, each qubit layer is indexed by $\hat{q} \wedge  \hat{z}$, and each $Z$-check layers is indexed by $\hat{x} \wedge \hat{q}$. Note the $\hat{x}$ and $\hat{z}$ axes are swapped here compared to the convention used in original Layer Code construction~\cite{williamson2023layer}

\begin{figure}[t]
	\centering
	\includegraphics[scale = .8, page=3]{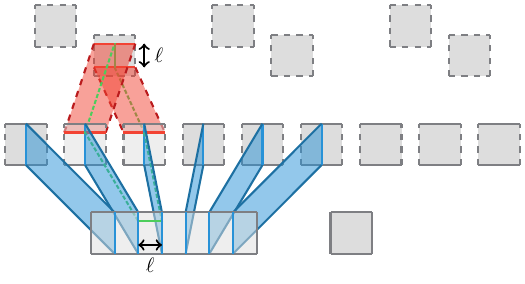}
    \caption{Interactions. $X$- (blue) and $Z$-check layers (red) are glued to the corresponding qubit layers by topological defects, and to one another along green string defects that ensure the output is a valid stabilizer code, i.e., all checks commute. The chosen graph coloring ensures the length $\ell$ of the green string defects are consistent. Extended boundary conditions are used for the $X$- and $Z$-check layers~\cite{williamson2025partial}. See Figure \ref{fig:shor-dimensions} for the final result for the Shor code.
    }
    \label{fig:shor-defect}
\end{figure}

We now describe the values that these coordinates take.  
We introduce a triple of integers $\chi_X,\chi_Q,\chi_Z,$ which are defined below, that specify the maximum coordinate in each direction.
Each of these quantities corresponds to the number of colors used in coloring a specific graph whose definition ensures that we obtain a consistent embedding.
The coordinate along $\hat{x}$ then takes values from $[\chi_X] = \{1,2,\dots, \chi_X\}$ and similarly for $\hat{q},\hat{z}$, as depicted in Fig. \ref{fig:shor-dimensions}.
Specifically, let $\sX,\sQ,\sZ$ denote the collections of $X$-checks, qubits, and $Z$-checks of the given input code $A$, respectively.
There is an associated graph structure, defined as follows.


\begin{figure}[t]
    \centering
    \centering
    \includegraphics[scale = .9, page=4]{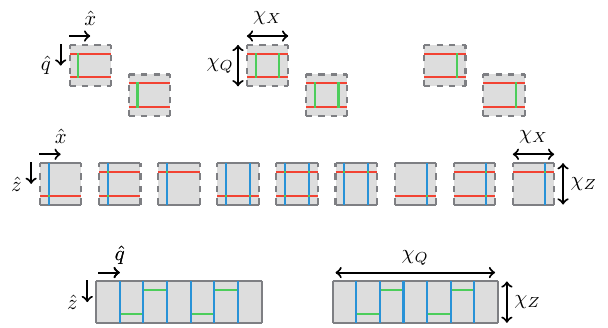}
    \caption{Dimensions of the weight reduction layers. Blue, red, and green lines denote topological defects along which pairs of layers are glued, see Fig.~\ref{fig:shor-defect}. Sizes of the layers are determined by the chromatic numbers of induced graphs. In this example, $\chi_X=\chi_Z=2,$ $\chi_Q=6$. For simplicity of presentation, the $\chi_Q$ direction is not drawn to scale.}
    \label{fig:shor-dimensions}
\end{figure}

\begin{definition}[Graph induced by $A$]
    \label{def:induced-graphs}
    Let ${G_X = (\sX,\sE_X)}$ denote the graph with vertices $\sX$ and edges $xx'\in \sE_{X}$ if the support of $x,x'$ with respect to code $A$ overlap or if there exists $z\in \sZ$ whose support overlaps with both that of $x$ and $x'$. Similarly, define $G_Z = (\sZ,\sE_Z)$.
    Let $G_Q = (\sQ,\sE_Q)$ denote the graph with vertices $\sQ$ and edges $qq'\in \sE_Q$ if there exists $x$ or $z$ check such that $q,q'$ are both in its support.
    Let $\chi_{\alpha},\alpha=X,Q,Z$ such that $G_\alpha$ is $\chi_\alpha$-colorable so that 
    \begin{align}
        \chi_X, \chi_Z &=O(\weight_X \qubit_X \weight_Z \qubit_Z)\\
        \chi_Q &=O(\weight_X\qubit_X+\weight_Z\qubit_Z)
    \end{align}
    where $\weight_X,\weight_Z$ are the maximum weights of the $X$-, $Z$-checks, while $\qubit_X,\qubit_Z$ are the maximum $X$-, $Z$-qubit degrees (the total qubit degree $\qubit \le \qubit_X+\qubit_Z$), respectively.
    The coloring scheme induces a mapping $\eta_X:\sX\to [\chi_X]$ and similarly for $Q,Z$.
    We may omit the subscript below, when the meaning is clear from context. 
\end{definition}

The appeal behind the use of graph colorings is that they spontaneously provide a labeling on the vertices of a sparse graph $\eta: \sV \rightarrow [\chi]$ that is consistent across $\sV$, despite $|\sV|$ being potentially much larger than $\chi$. The choice of graph in Definition~\ref{def:induced-graphs} was made specifically to ensure that the defects used to join layers together do not collide, and that the lengths of regions paired by each defect precisely match.
The full quantum weight reduction algorithm is described in Algorithm~\ref{alg:sparsification}.
This brings us to the statement of our main result.

\begin{theorem}[Algorithm~\ref{alg:sparsification}]
    \label{thm:main}
    Let $D$ be an $\llb n,k,d\rrb$ CSS code with maximum weight $\weight$, and total qubit degree $\qubit$, and integers $\chi_X,\chi_Q,\chi_Z \ge \chi$ such that $G_\alpha$ from Definition~\ref{def:induced-graphs} is $\chi_\alpha$-colorable; then Algorithm \ref{alg:sparsification} outputs a CSS code $D^\mathrm{sparse}$ with parameters $\llb O(\weight^4 \qubit^4) n, k, \Omega(\chi/\weight) d\rrb$, maximum weight $6$, and maximum total qubit degree $6$. Specifically, the maximum $X$- and $Z$-qubit degree is $4$.
\end{theorem}


\begin{remark}
    \label{rem:distance-enhancement}
    Algorithm~\ref{alg:sparsification} is valid for arbitrary chosen integers $\chi_X,\chi_Q,\chi_Z,$ provided that they are at least as large as the relevant chromatic numbers for each graph. 
    By choosing larger integers, the distance can be boosted to become arbitrarily large, potentially at the cost of additional qubit overhead.
    In particular, this implies that we can choose $\chi_X,\chi_Q,\chi_Z = \Theta(\weight^2 \qubit^2)$ and thus the distance of the sparsification is enhanced by a factor of $\Omega(\weight \qubit^2)$ with overhead $O(\weight^4 \qubit^4)$ in qubit size, yielding Theorem \ref{thm:main-intro}.  
    However, this conservative choice may not be necessary and may be due to an artifact of the proof (see Remark \eqref{rem:small-relative-expansion}).
    Specifically, we expect that it is sufficient for $\chi$ in Theorem~\ref{thm:main} to satisfy only $\chi_X,\chi_Z\geq \chi$.
    To emphasize the relation between the chromatic numbers and our construction, we henceforth assume that $\chi_X,\chi_Q,\chi_Z,$ denote the chromatic numbers of their respective graphs. 
    
\end{remark}


\begin{proof}
    Apply Algorithm \ref{alg:sparsification}, and the parameters of the output construction are obtain by identifying $D$ with $A$, and $D^\mathrm{sparse}$ with $C$ in Theorems~\ref{thm:logical} and~\ref{thm:distance}.
\end{proof}

\begin{figure}[t]
\begin{algorithm}[H]
	\caption{CSS code sparsification}
	\label{alg:sparsification}
	\begin{algorithmic}
		\Require An $\llb n,k,d\rrb$ CSS code $D$ with maximum weights $\weight$, and total qubit degree $\qubit$, and integers $\chi_X,\chi_Q,\chi_Z\ge \chi$ such that $G_\alpha$ is $\chi_\alpha$-colorable \eqref{def:induced-graphs}
        \Ensure An $\llb O(\weight^4 \qubit^4)n,k, \Omega(\chi/\weight) d \rrb$ CSS code $D^\mathrm{sparse}$ with maximum weight $6$, and total qubit degree 6.  Specifically, the maximum $X$- and $Z$-qubit degree is $4$.
		\ForEach {check $x \in \sX=\{1, \dots, n_Z\}$}
			\State Create a surface code layer of size $\chi_Q \times \chi_Z$
		\EndFor 
		\ForEach {qubit $q \in \sQ= \{1, \dots, n\}$}
			\State Create a surface code layer of size $\chi_X \times \chi_Z$
		\EndFor  
		\ForEach {check $z \in \sZ=\{1, \dots, n_Z\}$}
			\State Create a surface code layer of size $\chi_X \times \chi_Q$
		\EndFor  
        \ForEach {check $x \in \sX$}
			\ForEach {qubit $q$ in the support of $x$}
				\State 
                Identify the qubits along $(\eta(q),1) \rightarrow (\eta(q),\chi_Z) $ in $x$ layer with those along $(\eta(x),1) \rightarrow (\eta(x),\chi_Z)$ in the $q$ layer and implement a blue line defect, or $g^{QX}$ from Section \ref{sec:algebra}.
			\EndFor 
		\EndFor 
		\ForEach {check $z \in \sZ$}
			\ForEach {qubit $q$ in the support of $z$}
				\State 
                Identify the qubits along $(1,\eta(q)) \rightarrow (\chi_X,\eta(q)) $ in $z$ layer with those along $(1,\eta(z)) \rightarrow (\chi_X,\eta(z))$ in $q$ layer and implement a red line defect, or $g^{ZQ}$ from Section \ref{sec:algebra}.
			\EndFor 
		\EndFor 
		
		\ForEach {check $x \in \sX$}
			\ForEach {check $z \in \sZ$}
				\State Let $O \subseteq \sQ$ be the overlap of the support of $x$ and $z$.
				\State Let $\eta(O) = (\eta_1, \eta_2, \dots, \eta_{2t})$ be the coordinates associated with $O$ in increasing order.
				\ForEach { $i\in [1, \dots, t]$}
					\State Add a green defect, or $p^{ZX}$ from Section \ref{sec:algebra}, along $(\eta_{2i-1}, \eta(z)) \rightarrow (\eta_{2i}, \eta(z))$ on the $x$ layer, and $(\eta_{2i-1}, \eta(x)) \rightarrow (\eta_{2i}, \eta(x))$ on the $z$ layer.
				\EndFor
			\EndFor 
		\EndFor 
	\end{algorithmic}
\end{algorithm}
\end{figure}
\section{Preliminaries}
\label{sec:prelim}

In this section we introduce concepts from the theory of codes and chain complexes that form the background for our Layer Code quantum weight reduction procedure. 

\begin{definition}
\label{def:chain-complex}
A \textbf{complex} $C$ is a sequence of $\dF_2$-vector spaces $C_{i}$ together with linear $\partial_{i}:C_{i} \to C_{i-1}$, called the \textbf{differentials} of $C$, such that $\partial_{i}\partial_{i+1}=0$ where the subscripts are often omitted. We write
\begin{equation}
    C = \cdots \to  C_{i} \xrightarrow{\partial_{i}} C_{i-1} \to \cdots
\end{equation}
Note that $\im \partial_{i+1} \subseteq \ker \partial_i$ for all $i$, and thus the \textbf{$i$-homology} of $C$ is defined as
\begin{equation}
    H_i(C)\equiv  \ker \partial_i/\im \partial_{i+1}
\end{equation}
Denote the equivalences class of $\ell_{i} \in C_{i}$ as $[\ell_{i}]\in H_i(C)$ -- conversely, we may also write $[\ell_i] \in H_i(C)$ without specifying the representation $\ell_i$.
\end{definition}

\begin{definition}[Basis]
A complex $C$ with \textbf{basis} is such that each $C_i$ is equipped with a canonical basis whose elements are called \textbf{$i$-cells}. This further induces a well-defined nondegenerate bilinear form $\bra \cdot|\cdot\ket$ and \textbf{(Hamming) weight} $|\cdot|$.
We say that cells $c_{i},c_{i-1} $ are \textbf{adjacent}  if $\bra c_{i-1} |\partial c_{i}\ket \ne 0$, and write $c_i \sim c_{i-1}$. 
\end{definition}

\begin{example}(Repetition Code)
\label{ex:rep}
The \textbf{repetition code} on $L$ qubits is the classical code with parity checks $Z_i Z_{i+1}$ where $i=1,...,L-1$.
Its corresponding complex is defined as $R \equiv R(L)$ with differential $\partial^{R}$.
The 1-cells are denoted via $|i^+\ket$ for $i=1,...,L-1$ where $i^\pm =i\pm 1/2$, and 0-cells via $|i\ket$ for $i=1,...,L$ so that
\begin{equation}
    \partial^{R} |i^+\ket = |i\ket +|i+1\ket
\end{equation}
We implicitly assume that $|i\ket, |i^\pm\ket =0$ if the label is not within the previous parameters.
\end{example}

A definition of the classical repetition code is necessary since the surface code (with possible boundary conditions) will be the homological product of two repetition codes \cite{tillich2013quantum}.
\begin{lemma}[Repetition Code]
\label{lem:rep}
Let $R=R(L)$ denote the repetition code. Then $H_1(R) =0$ and $H_0(R) \cong \dF_2$. The unique basis element of $H_0(R)$ is given by $[|i\ket]$ for any $i=1,...,L$ and the unique basis element of $H^0(R)$ is $[\sR_0]$ where
\begin{equation}
    \sR_0 \equiv \sum_{i=1}^{L} |i\ket
\end{equation}
Note that $\sR_0$ can be regarded as the collection of all 0-cells.
\end{lemma}

\begin{definition}[CSS Codes]
\label{def:CSS}
Consider a CSS code as a complex $C=C_2\to C_1 \to C_0$ (with basis) and convention that $C_2,C_1,C_0$ are the $X$-type checks, qubits, $Z$-type checks, respectively. Let $\weight_X,\weight_Z$ denote the max weight of $X,Z$-checks, respectively, and let $\qubit_X,\qubit_Z$ denote the max number of $X,Z$-checks acting on any qubit, respectively.
Note that $\weight_X,\qubit_X$ are the max column weights $\norm{\cdot}_{\rm{col}}$ of $\partial_2,\partial_1$, respectively, and similarly, $\qubit_Z,\weight_Z$. 
Hence, we denote the weights via the following diagram
\begin{equation}
\label{eq:weight-diagram}
C_2 \xrightleftharpoons[\qubit_X]{\weight_X} C_1 \xrightleftharpoons[\weight_Z]{\qubit_Z} C_0
\end{equation}
Furthermore, given cells $c_2,c_0$, we shall refer to the \textbf{common qubits} -- denoted as $c_2\wedge c_0$ -- as intersection of supports of $c_2,c_0$, i.e., the collection of cells $c_1$ such that $c_2 \sim c_1 \sim c_0$.
\end{definition}

\begin{definition}[Framework]
\label{def:framework}
The following diagram summarizes the framework introduced in Theorem I.1 of Ref.~\cite{yuan2025unified}, albeit with slightly modified notation (subscripts of maps may be omitted for simplicity).
\begin{equation}
\label{eq:height-2-diagram}
\begin{tikzpicture}[baseline]
\matrix(a)[matrix of math nodes, nodes in empty cells, nodes={minimum size=25pt},
row sep=2em, column sep=2em,
text height=1.25ex, text depth=0.25ex]
{&& \red{C^{X}_{2}}  & \red{C^{X}_{1}} & \red{C^{X}_{0}}\\
& \gray{C^{Q}_{2}}  & \gray{C^{Q}_{1}}  & \gray{C^{Q}_{0}} &\\
\blue{C^{Z}_{2}} & \blue{C^{Z}_{1}} & \blue{C^{Z}_{0}} &&\\};
\path[->,red,font=\scriptsize]
(a-1-3) edge node[above]{$\partial^{X}_2$} (a-1-4)
(a-1-4) edge node[above]{$\partial^{X}_1$} (a-1-5);
\path[->,gray,font=\scriptsize]
(a-2-2) edge node[above]{$\partial^{Q}_2$} (a-2-3)
(a-2-3) edge node[above]{$\partial^{Q}_1$} (a-2-4);
\path[->,blue,font=\scriptsize]
(a-3-1) edge node[above]{$\partial^{Z}_2$} (a-3-2)
(a-3-2) edge node[above]{$\partial^{Z}_1$} (a-3-3);
\path[->,font=\scriptsize]
(a-1-3) edge[bend right=80, dashed] node[left]{$p^{ZX}_2$} (a-3-2)
(a-1-4) edge[bend left=80, dashed] node[right]{$p^{ZX}_1$} (a-3-3);
\path[->,font=\scriptsize]
(a-1-3) edge node[right]{$g^{QX}_2$} (a-2-3)
(a-1-4) edge node[right]{$g^{QX}_1$} (a-2-4)
(a-2-2) edge node[right]{$g^{ZQ}_2$} (a-3-2)
(a-2-3) edge node[right]{$g^{ZQ}_1$} (a-3-3);
\end{tikzpicture}
\end{equation}

We say that $g^{QX},g^{ZQ},p^{ZX}$ are \textbf{compatible} with complexes $C^X,C^Q,C^Z$ if 
\begin{align}
    \partial^{Q} g^{QX} &= g^{QX} \partial^X \\
    \partial^{Z} g^{ZQ} &= g^{ZQ} \partial^Q \\
    \label{eq:chain-homotopy}
    g^{ZQ} g^{QX} &= \partial^{Z}p^{ZX} +p^{ZX} \partial^{X}
\end{align}
Since we will also be concerned with the weights of the constructed complex $C$, it's visually more straightforward to use the following diagram, where each arrow corresponds to the max column weight of the corresponding linear map.
\begin{equation}
\label{eq:weight-diagram-2D}
\begin{tikzpicture}[baseline]
\matrix(a)[matrix of math nodes, nodes in empty cells, nodes={minimum size=25pt},
row sep=2em, column sep=2em,
text height=1.25ex, text depth=0.25ex]
{&& \red{C^{X}_{2}}  & \red{C^{X}_{1}} & \red{C^{X}_{0}}\\
& \gray{C^{Q}_{2}}  & \gray{C^{Q}_{1}}  & \gray{C^{Q}_{0}} &\\
\blue{C^{Z}_{2}} & \blue{C^{Z}_{1}} & \blue{C^{Z}_{0}} &&\\};
\path[-left to,font=\scriptsize,transform canvas={yshift=0.2ex}]
(a-1-3) edge node[above]{}  (a-1-4)
(a-1-4) edge node[above]{}  (a-1-5)
(a-2-2) edge node[above]{}  (a-2-3)
(a-2-3) edge node[above]{}  (a-2-4)
(a-3-1) edge node[above]{}  (a-3-2)
(a-3-2) edge node[above]{}  (a-3-3);
\path[left to-,font=\scriptsize,transform canvas={yshift=-0.2ex}]
(a-1-3) edge node[below]{}  (a-1-4)
(a-1-4) edge node[below]{}  (a-1-5)
(a-2-2) edge node[below]{}  (a-2-3)
(a-2-3) edge node[below]{}  (a-2-4)
(a-3-1) edge node[below]{}  (a-3-2)
(a-3-2) edge node[below]{}  (a-3-3);
\path[-left to,font=\scriptsize,transform canvas={xshift=0.2ex}]
(a-1-3) edge node[right]{}  (a-2-3)
(a-1-4) edge node[right]{}  (a-2-4)
(a-2-2) edge node[right]{}  (a-3-2)
(a-2-3) edge node[right]{}  (a-3-3);
\path[left to-,font=\scriptsize,transform canvas={xshift=-0.2ex}]
(a-1-3) edge node[left]{}  (a-2-3)
(a-1-4) edge node[left]{}  (a-2-4)
(a-2-2) edge node[left]{}  (a-3-2)
(a-2-3) edge node[left]{}  (a-3-3);
\path[-left to,font=\scriptsize]
(a-1-3) edge[bend right=80, dashed] node[right]{} (a-3-2)
(a-1-4) edge[bend left=80, dashed] node[right]{} (a-3-3);
\path[left to-,font=\scriptsize]
(a-1-3) edge[bend right=80, dashed] node[left]{} (a-3-2)
(a-1-4) edge[bend left=80, dashed] node[left]{} (a-3-3);
\end{tikzpicture}
\end{equation}
In particular, an upper bound of the column/row weights of $C$ of diagram in Eq.~\eqref{eq:weight-diagram} can be obtained via the expanded diagram above in a straightforward (but notationally tedious) manner.
See, for example, Eq.~\eqref{eq:reduced-weights} and Eq.~\eqref{eq:weight-diagram-2D-layer}.
\end{definition}

\begin{theorem}[Theorem I.1 of Ref.~\cite{yuan2025unified}]
    \label{thm:unified}
    If the maps in Eq.~\eqref{eq:height-2-diagram} are compatible, then $C=C_2\to C_1\to C_0$ constructed via $C_i = C_i^X\oplus C_i^Q \oplus C_i^Z$ and differential
    \begin{equation}
        \partial_i = 
        \begin{pmatrix}
            \partial^{X}_i &  &\\
            g^{QX}_i & \partial^{Q}_i & \\
            p^{ZX}_i & g^{ZQ}_i & \partial^{Z}_i
        \end{pmatrix}
    \end{equation}
    is a well-defined complex. Moreover, the following -- referred as the \textbf{embedded} complex -- is also well-defined 
    \begin{equation}
        C^{\bg} = H_2(C^X) \xrightarrow{[g^{QX}_2]} H_1(C^Q) \xrightarrow{[g^{ZQ}_1]} H_0 (C^Z)
    \end{equation}
    where $[g^{QX}],[g^{ZQ}]$ are defined via passing through quotients. Moreover, if $H_1(C^X) = H_1(C^Z) =0$, then 
    \begin{equation}
        H_1(C) \cong H_1(C^{\bg})
    \end{equation}
    Specifically, the isomorphism $H_1(C) \to H_1 (C^{\bg})$ is given by
    \begin{equation}
        \label{eq:height-2-cone-iso}
        \left[
        \begin{pmatrix} 
        0 \\ \ell^{Q} \\ \ell^{Z}
        \end{pmatrix} \right]
        \mapsto 
        \llb \ell^{Q} \rrb
    \end{equation}
    which states that every logical of $C$ has an equivalent representation $\ell=\ell^\qubit \oplus \ell^Z$ which is not supported on $C_1^X$, the projection of which onto $\ell^Q \in C_1^{Q}$ is a logical of $C^{Q}$ so that the equivalence class $[\ell^Q]\in H_1(C^Q)$ is a logical of $C^{\bg}$ and thus $\llb \ell^Q \rrb \in H_1(C^{\bg})$.
\end{theorem}

The following provides a specific example which illustrates Theorem~\ref{thm:unified} in practice.

\begin{example}[Euclidean Layer Code]
    Consider the Euclidean layer code $C$ constructed with input code $A$ on 3 qubits with parity checks $XXI$ and $ZZZ$.
    As illustrated in Fig. \ref{fig:example-layer-code}, the red, grey, and blue layer(s) denote those corresponding to $X$-check(s), qubits, and $Z$-check(s), respectively.
    We then claim that the squiggly lines denotes a general logical representation that corresponds to the logical $XIX$ of the input code $A$, in the sense of the isomorphism in Theorem~\ref{thm:unified}.
    Specifically, as illustrated by Fig. \ref{fig:example-isomorphism}, given a general logical representation (top left), it can always be \textit{cleaned} by applying parity checks, e.g., those corresponding to the oddly shaped red region, to an equivalent logical representation $\ell$  (top right) which is not supported on $C_1^X$. 
    The restriction (bottom right) to $\ell^Q \in C_1^Q$ are logical representations of $C^Q$ and thus its equivalence class (bottom left) $[\ell^Q] \in H_1(C^Q) = A_1$ corresponds to the logical representation $XIX$ so that its further equivalence class $\llb \ell^Q\rrb \in H_1(A)$.
    \begin{figure}[ht]
    \centering
    \subfloat[\label{fig:example-layer-code}]{%
        \centering
        \includegraphics[width=0.8\columnwidth]{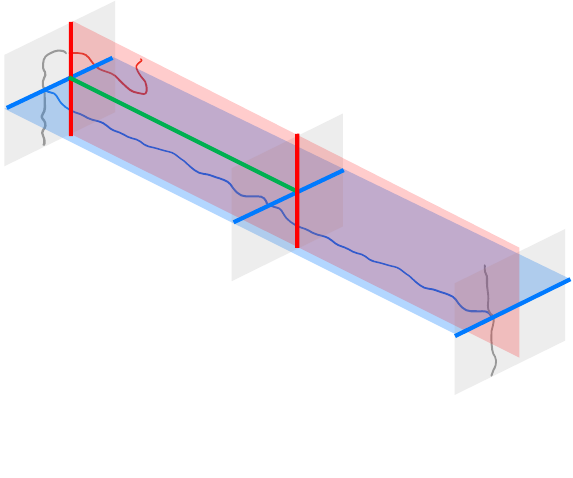}
    }
    \\
    \subfloat[\label{fig:example-isomorphism}]{%
        \centering
        \includegraphics[width=1\columnwidth]{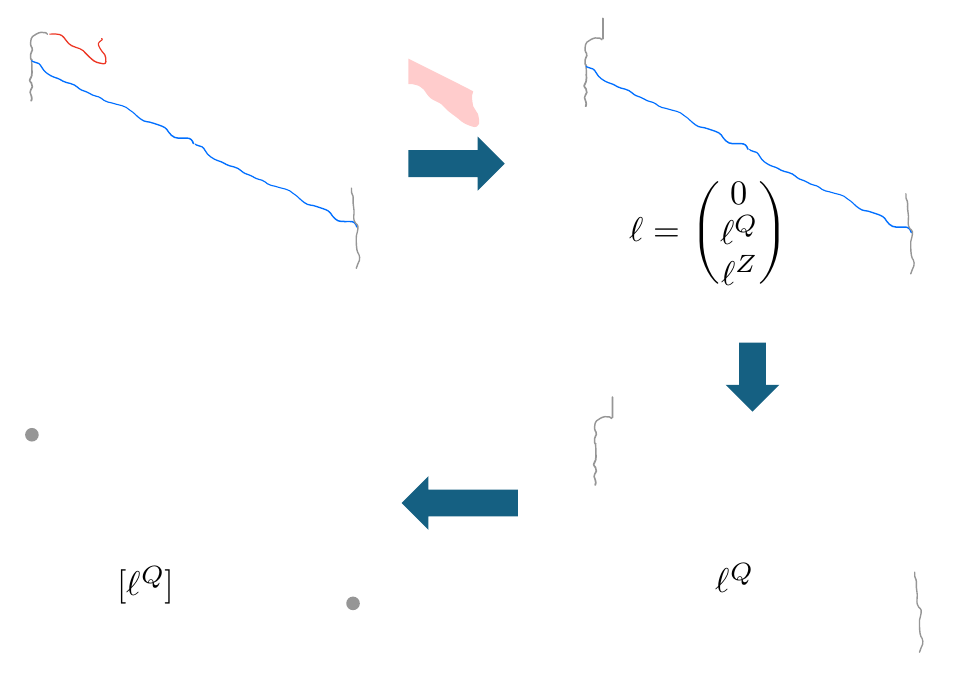}
    }
    \caption{Layer Code Example. (a) denotes the Euclidean layer code $C$ constructed from input code $A$ with parity checks $XXI,ZZZ$, where the squiggly line denotes a possible logical representation (b) illustrates the isomorphism in Theorem~\ref{thm:unified} using the logical representation in (a) as an example. In particular, it corresponds to the logical $XIX$ of the input code $A$.
    }
    \label{fig:example-layer}
    \end{figure}
\end{example}

\section{Algebraic Construction}
\label{sec:algebra}

In this section we give an algebraic description of our Layer Code quantum weight reduction procedure. 

Let us start with a CSS code with corresponding complex $A =X\to Q \to Z$ with cells $x \in \sX,q \in \sQ,z \in \sZ$, respectively, and weight
\begin{equation}
\label{eq:weight-diagram-input}
X \xrightleftharpoons[\qubit_X]{\weight_X} Q \xrightleftharpoons[\weight_Z]{\qubit_Z} Z
\end{equation}
There is then an associated graph structure as defined in Definition~\ref{def:induced-graphs}.

\begin{remark}[Induced Graphs]
    The induced graphs and coloring scheme in Definition~\ref{def:induced-graphs} are somewhat necessary in the following sense. Suppose that a mapping $\eta_X:\sX \to \dZ$ is given and similarly for $Q,Z$ so that the maps are defined as in Eq.~\eqref{eq:gQX}-\eqref{eq:pZX}, then it's sufficient to find $\eta$ which is injective on 
    \begin{enumerate}[label=\arabic*)]
        \item $x\wedge z$ for any pair -- so that $p^{ZX}$ in Eq.~\eqref{eq:pZX} is defined
        \item $\{q:q\sim x\}$ for any $x$ and $\{q:q\sim z\}$  -- so that $\|g^{QX}\|_{\rm{col}},\|g^{ZQ}\|_{\rm{row}}\le 1$ in Eq.~\eqref{eq:gQX-col}
        \item $\{z:z\sim q\}$ and $\{x:x\sim q\}$ for any $q$  -- so that $\|g^{ZQ}\|_{\rm{col}},\|g^{QX}\|_{\rm{row}}\le 1$ in Eq.~\eqref{eq:gZQ-col}
        \item $\{z: x\wedge z\ne \emptyset\}$ for any $x$ and $\{x: x\wedge z\ne \emptyset\}$ for any $z$ -- so that $\|p^{ZX}\|_{\rm{col/row}}\le 1$ in Eq.~\eqref{eq:pZX-col}
    \end{enumerate}
    If $\eta$ satisfies the conditions above, then $\eta$ is necessarily the coloring schemes of the induced graphs.
\end{remark}

Using the induced coloring scheme $\eta$, we can define independent CSS codes that will replace the checks and qubits of the input code $A$. 
Specifically, define $R_X =R(\chi_X),R_Q=R(\chi_Q),R_Z=R(\chi_Z)$ as repetition codes with 0-cells denoted by $|i\ket,|j\ket,|k\ket$, respectively.
Define the independent CSS codes as
\begin{align}
    \label{eq:CX}
    C^X &= X\otimes R_Q^\top \otimes R_Z^\top \\
    C^Q &= R_X\otimes Q \otimes R_Z^\top \\
    \label{eq:CZ}
    C^Z &= R_X\otimes R_Q \otimes Z
\end{align}
where we treat $X,Q,Z$ as length-0 complexes so that, e.g., $C^X$ can be regarded as $|\sX|$ copies of the length-2 complex $R_Q^\top \otimes R_Z^\top$ with 2-cells $|x,j,k\ket$ where $x\in \sX,j\in [\chi_Q],k\in [\chi_Z]$. 
The case is similar for $C^Q,C^Z$.

Using the framework in Ref. \cite{yuan2025unified}, summarized in Definition~\ref{def:framework}, define the following
\begin{align}
    \label{eq:gQX}
    g^{QX} |x,j,k_s\ket &=\sum_{q} |\eta(x), q,k_s\ket 1\{j=\eta(q),q\sim x\} \\
    g^{ZQ} |i_s,q,k\ket &=\sum_{z} |i_s,\eta(q),z\ket 1\{k=\eta(z),z\sim q\}
\end{align}
where $i_s,j_s,k_s$ denote integers, half-integers for $s=0,1$, respectively. 

Further define $p^{ZX}$ as follows.
Given $x,z$ checks, there are an even number of common qubits and we note that $q_1,\dots,q_{2t} \in x\wedge z$ can be ordered in a strictly increasing manner with respect to $\eta$, i.e., $\eta(q_1)<\cdots <\eta(q_{2t})$, since $q_1,…,q_{2t}$ are all adjacent to each other in the induced graph $(\sQ,\sE_{Q})$. 
We can then define the \textit{string defect} $\Gamma(x\wedge z)\subseteq \dR$ as the disjoint union of half-open intervals $[\eta(q_{2i-1}),\eta(q_{2i}))$ over $i=1,...,t$. Define the following
\begin{align}
    \label{eq:pZX}
    p^{ZX} |x,j_s,k\ket &=\sum_{z} 1\{k=\eta(z)\} |\eta(x),j_s^+,z\ket \\
    &\quad\quad\times1\{j_s\in \Gamma(x\wedge z)\} \nonumber
\end{align}

\begin{theorem}[Logical Subspace]
    \label{thm:logical}
    The maps $g^{QX},g^{ZQ},p^{ZX}$ defined in Eq.~\eqref{eq:gQX}-\eqref{eq:pZX} are compatible with complexes $C^X,C^Q,C^Z$ in Eq.~\eqref{eq:CX}-\eqref{eq:CZ}.
    Moreover, the embedded complex in the constructed complex $C$ is exactly $=A$  so that $H_1(C) \cong H_1(A)$, and has weights
    \begin{equation}
    \label{eq:reduced-weights}
    C_2 \xrightleftharpoons[4]{6} C_1 \xrightleftharpoons[6]{4} C_0
    \end{equation}
    And dimension
    \begin{equation}
        \dim C_i = O\left(\sum_{\alpha=X,Q,Z} \dim \alpha |\chi_\beta||\chi_\zeta|\right)
    \end{equation}
    where $\beta,\zeta$ are the remaining indices other than $\alpha$.
    In particular, as a CSS code, $C$ has qubit overhead $O(\weight^4\qubit^4)$ where $\weight,\qubit$ are the maximum of $\weight_X,\weight_Z$ and $\qubit_X,\qubit_Z$, respectively.
\end{theorem}

\begin{theorem}[Code Distance]
    \label{thm:distance}
    The constructed complex $C$ in Theorem~\ref{thm:logical} has $\alpha=X,Z$-type code distance satisfying
    \begin{equation}
        d_\alpha(C) \ge \min(\chi_{\bar{\alpha}}, 2\chi_{Q}) \frac{1}{\weight_\alpha} d_\alpha(A)
    \end{equation}
    where $\bar{\alpha}$ is the complement index of $\alpha$ and $d_\alpha(A)$ is the $\alpha$-type distance of $A$.
\end{theorem}

\begin{remark}
    \label{rem:small-relative-expansion}
    As discussed in Definition~\ref{def:induced-graphs}, it's typical to expect
    \begin{align}
        \chi_{X},\chi_{Z} &= \Theta(\weight_X\qubit_X\weight_Z \qubit_Z) \\
        \chi_{Q} &= \Theta(\weight_X\qubit_X +\weight_Z \qubit_Z) 
    \end{align}
    Hence, for large weights $\weight_X,\weight_Z=\Theta(\weight) $ and qubit degrees $\qubit_X,\qubit_Z=\Theta(\qubit)$, we have $\chi_{Q}\ll \chi_{X},\chi_{Z}$ and thus the code distance is
    \begin{equation}
        d(C) \ge \Theta(\qubit) d(A)
    \end{equation}
    Note that this lower bound in code distance is not ideal as discussed in Remark~\ref{rem:distance-enhancement} due to $\chi_{Q} \ll \chi_{X},\chi_{Z}$.
    
    However, this lower bound may be an artifact of our proof.
    Specifically, when $\chi_{Q} \ll \chi_{X},\chi_{Z}$, the relative expansion coefficient $\alpha \ll 1$ as shown in Proposition \eqref{prop:square-expansion}. In this case, if we refer to Eq. (A26)-(A27) of \cite{yuan2025unified}, we see that the inequality applied in the Cleaning Lemma is very non-optimal, despite being necessary to apply the triangle inequality in Eq. (A28) of \cite{yuan2025unified}.
\end{remark}
\begin{remark}[Technical Detail]
    In our proof of the code distance, we do not use the detailed structure of the defect map $p^{ZX}$. 
In fact, as long as the defect map is compatible with the gluing maps $g^{ZQ},g^{QX}$ in the sense of Eq.~\eqref{eq:chain-homotopy}, Theorem~\ref{thm:logical} and~\ref{thm:distance} still hold, albeit with distinct reduced weights (not necessarily 6 and 4 in Eq.~\eqref{eq:reduced-weights}).
    Intuitively, the consistency conditions are sufficient for $p^{ZX}$ to not introduce new (small) logicals, while the logicals originating from $A$ are ``pinned" due to the relative expansion.
\end{remark}

\section{Proofs of the main theorems}
\label{sec:proof}

In this section we provide detailed proofs of our main results. 

\subsection{Proof of Theorem \eqref{thm:logical}}
\label{sec:proof-logical}
\begin{proof}[Part I]
    Let us first check that the gluing and defect maps satisfy the compatibility conditions.
    Indeed,
    \begin{align}
        g^{QX} \partial^{X} |xjk\ket &= g^{QX} \sum_{s=\pm} |x j k^s\ket \\
        &= \sum_{s=\pm} \sum_{q:q\sim x} |\eta(x)q k^s\ket 1\{j=\eta(q)\} \\
        &= \partial^{Q} g^{QX} |x,j,k\ket 
    \end{align}
    Hence, $g^{QX}$ is a chain map. The case is similar for $g^{ZQ}$.
    Note that
    \begin{align}
        g^{ZQ} g^{QX} |x j k\ket &= g^{ZQ} \sum_{q:q\sim x} |\eta(x) q k\ket 1\{j=\eta(q)\} \\
        &= \sum_{z,q:z\sim q\sim x} |\eta(x),\eta(q),z\ket  \\
        &\quad\quad \times1\{j=\eta(q),k=\eta(z)\} \nonumber\\
        &=\sum_{z} |\eta(x),j,z\ket 1\{k=\eta(z)\} \\
        &\quad\quad \times \sum_{q} 1\{j=\eta(q),q\in x\wedge z\} \nonumber
    \end{align}
    By the induced graph $(\sQ,\sE_Q)$, we see that $\eta$ is injective on $x\wedge z$ and thus
    \begin{align}
        g^{ZQ} g^{QX} |x,j,k\ket &=\sum_{z} |\eta(x),j,z\ket \\
        &\quad 1\{k=\eta(z)\} 1\{j \in \eta(x\wedge z)\}
    \end{align}
    Also note that
    \begin{align}
        \partial^Z p^{ZX} |x,j,k\ket &= \sum_{z} 1\{k=\eta(z)\} 1\{j\in \Gamma(x\wedge z)\} \nonumber\\
        &\times (|\eta(x),j,z\ket +|\eta(x),j+1,z\ket) 
    \end{align}
    And that
    \begin{align}
        p^{ZX} \partial^X |x j k\ket &= \sum_{z} 1\{k=\eta(z)\} \\
        &\quad\;(|\eta(x) j z\ket 1\{j^-\in \Gamma (x\wedge z)\} \nonumber\\
        &\quad\; +|\eta(x),j+1,z\ket 1\{j^+\in \Gamma (x,z)\}) \nonumber
    \end{align}
    Note that by definition
    \begin{align}
        1\{j\in \Gamma(x\wedge z) \} + 1\{j^-\in \Gamma(x\wedge z)\} &= 1\{j\in \eta(x\wedge z)\} \nonumber \\
        1\{j\in \Gamma(x\wedge z)\} + 1\{j^+\in \Gamma(x\wedge z)\} &=0
    \end{align}
    Therefore, $g^{ZQ}g^{QX} = \partial^{Z} p^{ZX} + p^{ZX} \partial^{X}$
\end{proof}

\begin{proof}[Part II]
    Let us next check that the embedded complex is exactly $A$.
    By the K\"{u}nneth Formula and Lemma \eqref{lem:rep}, $H_2(C^X) \cong X$ and has basis $[|x\ket],x\in \sX$ with unique representation
    \begin{equation}
        |x\ket = \sum_{j,k} |x,j,k\ket
    \end{equation}
    Similarly, $H_1(C^Q)\cong Q$ and has basis $[|y\ket],y\in \sQ$ with possible (but not comprehensive) representation
    \begin{equation}
        \label{eq:H1CQ-basis}
        |q\ket = \sum_k |i,q,k\ket
    \end{equation}
    where $i\in [\chi_X]$ is arbitrary.
    Similarly, $H_0(C^Z) \cong Z$ and has basis $[|z\ket],z\in \sZ$ with possible representation
    \begin{equation}
        |z\ket = |i,j,z\ket
    \end{equation}
    where $i\in [\chi_X],j\in [\chi_Q]$ are arbitrary. Note that
    \begin{align}
        [g^{QX}][|x\ket] &= [g^{QX} |x\ket] \\
        &=\left[ \sum_{j,k,q} |\eta(x),q,k\ket 1\{j=\eta(q),q\sim x\}\right] \\
        &= \sum_{y} 1\{q\sim x\} [|q\ket]
    \end{align}
    And that
    \begin{align}
        [g^{ZX}][|q\ket] &= \left[g^{QX} \sum_{k} |i,q,k\ket\right] \\
        &=\left[ \sum_{k,z} |i,\eta(q),z\ket 1\{k=\eta(z),z\sim q\}\right] \\
        &=\sum_{z} [|i,\eta(q),z\ket]1\{z\sim q\} \\
        &= \sum_{z} 1\{z\sim q\} [|z\ket]
    \end{align}
    Hence, the embedded complex (with the naturally induced basis $[|x\ket],[|y\ket],[|z\ket]$) is exactly $=A$.
    Since $C^X,C^Z$ do not have internal logicals, i.e., $H_1(C^X)=H_1(C^Z)$, we see that $H_1(C) \cong H_1(A)$.
\end{proof}

\begin{proof}[Part III]
Finally, let us prove the weights of the constructed complex $C$.
By the framework summarized in Definition \eqref{def:framework}, we investigate the maximum column and row weight of the differential
\begin{equation}
\partial = 
\begin{pmatrix}
    \partial^{X} &  &\\
    g^{QX} & \partial^{Q} & \\
    p^{ZX} & g^{ZQ} & \partial^{Z}
\end{pmatrix}
\end{equation}
The weights of the matrix elements are summarized via the following diagram
\begin{equation}
\label{eq:weight-diagram-2D-layer}
\begin{tikzpicture}[baseline]
\matrix(a)[matrix of math nodes, nodes in empty cells, nodes={minimum size=25pt},
row sep=2em, column sep=2em,
text height=1.25ex, text depth=0.25ex]
{&& \red{C^{X}_{2}}  & \red{C^{X}_{1}} & \red{C^{X}_{0}}\\
& \gray{C^{Q}_{2}}  & \gray{C^{Q}_{1}}  & \gray{C^{Q}_{0}} &\\
\blue{C^{Z}_{2}} & \blue{C^{Z}_{1}} & \blue{C^{Z}_{0}} &&\\};
\path[-left to,font=\scriptsize,transform canvas={yshift=0.2ex}]
(a-1-3) edge node[above]{$4$}  (a-1-4)
(a-1-4) edge node[above]{$2$}  (a-1-5)
(a-2-2) edge node[above]{$4$}  (a-2-3)
(a-2-3) edge node[above]{$2$}  (a-2-4)
(a-3-1) edge node[above]{$4$}  (a-3-2)
(a-3-2) edge node[above]{$2$}  (a-3-3);
\path[left to-,font=\scriptsize,transform canvas={yshift=-0.2ex}]
(a-1-3) edge node[below]{$2$}  (a-1-4)
(a-1-4) edge node[below]{$4$}  (a-1-5)
(a-2-2) edge node[below]{$2$}  (a-2-3)
(a-2-3) edge node[below]{$4$}  (a-2-4)
(a-3-1) edge node[below]{$2$}  (a-3-2)
(a-3-2) edge node[below]{$4$}  (a-3-3);
\path[-left to,font=\scriptsize,transform canvas={xshift=0.2ex}]
(a-1-3) edge node[right]{$1$}  (a-2-3)
(a-1-4) edge node[right]{$1$}  (a-2-4)
(a-2-2) edge node[right]{$1$}  (a-3-2)
(a-2-3) edge node[right]{$1$}  (a-3-3);
\path[left to-,font=\scriptsize,transform canvas={xshift=-0.2ex}]
(a-1-3) edge node[left]{$1$}  (a-2-3)
(a-1-4) edge node[left]{$1$}  (a-2-4)
(a-2-2) edge node[left]{$1$}  (a-3-2)
(a-2-3) edge node[left]{$1$}  (a-3-3);
\path[-left to,font=\scriptsize]
(a-1-3) edge[bend right=80, dashed] node[right]{$1$} (a-3-2)
(a-1-4) edge[bend left=80, dashed] node[right]{$1$} (a-3-3);
\path[left to-,font=\scriptsize]
(a-1-3) edge[bend right=80, dashed] node[left]{$1$} (a-3-2)
(a-1-4) edge[bend left=80, dashed] node[left]{$1$} (a-3-3);
\end{tikzpicture}
\end{equation}
Specifically, for any $\alpha=X,Q,Z$, we have
\begin{equation}
    \norm{\partial^{\alpha}_2}_{\rm{col}} \le 4, \quad \norm{\partial^{\alpha}_1}_{\rm{col}} \le 2
\end{equation}
Note that
\begin{equation}
    \label{eq:gQX-col}
    |g^{QX} |x,j,k_s\ket| =\sum_{q} 1\{j=\eta(q),q\sim x\}
\end{equation}
By the induced graph $(\sQ,\sE_Q)$, we see that $\eta$ is injective on $\{q:q\sim x\}$ and thus $\|g^{QX}_i\|_{\rm{col}} \le 1$ for $i=2,1$.
Similarly, note that
\begin{equation}
    \label{eq:gZQ-col}
    |g^{ZQ} |i_s,q,k\ket| =\sum_{z} 1\{k=\eta(z),z\sim q\}
\end{equation}
By the induced graph $(\sZ,\sE_Z)$, we see that $\eta$ is injective on $\{z:z\sim y\}$ and thus $\|g^{ZQ}_i\|_{\rm{col}} \le 1$ for $i=2,1$.
Finally, note that
\begin{equation}
    \label{eq:pZX-col}
    |p^{ZX} |x,j_s,k\ket| =\sum_{z} 1\{k=\eta(z)\} 1\{j_s\in \Gamma(x\wedge z)\}
\end{equation}
By the induced graph $(\sZ,\sE_Z)$, we see that for a given $x$, the coloring map $\eta$ is injective on the collection of $z$ which has support that nontrivially overlaps with that of $x$, i.e., $x\wedge z\ne \emptyset$ and thus $\|p^{ZX}_i\|_{\rm{col}} \le 1$ for $i=2,1$.
Therefore,
\begin{equation}
    \norm{\partial_2}_{\rm{col}} \le 6, \quad \norm{\partial_1}_{\rm{col}}\le 4
\end{equation}
The maximum row weight is similarly derived.
\end{proof}

\subsection{Proof of Theorem \eqref{thm:distance}}
\label{sec:proof-distance}

\begin{proof}
    Similar to Theorem IV.2 of Ref. \cite{yuan2025unified}, one can prove the statement by modifying certain aspects of the Cleaning Lemma in Ref. \cite{yuan2025unified}.
    However, in this manuscript, we provide an alternative route that was briefly commented subsequent to Theorem IV.2.
    Specifically, we can prove that $(\partial^X,g^{QX})$ satisfies the isoperimetric inequality/relative expansion property\footnote{See also Ref. \cite{baspin2025fast}, in which the property was referred as systolic expansion}, and thus apply the Cleaning Lemma as a black-box.
    In Ref. \cite{yuan2025unified}, the former method was chosen to optimize the $\Theta(1)$ factor in the distance lower bound to match that which originally appeared in the conventional layer code construction \cite{williamson2023layer}.
    However, since it's only a $\Theta(1)$ factor\footnote{In fact, the former method is only better than the latter by a factor of 2.} independent of $\weight,\qubit$, it may be beneficial to see how the alternative method connects nicely with the Cleaning Lemma.
    
    Without loss of generality, we focus on $d_X$ since the construction is symmetric in $X$ and $Z$.
    By Proposition \eqref{prop:square-expansion}, we see that $(\partial^{X},g^{QX})$ satisfies the isoperimetic inequality/relative expansion, i.e., for every $s^{X} \in C_2^{X}$, there exists $\hat{s}^{X}\in C_2^{X}$ with $\partial ^X s^X = \partial^X \hat{s}^{X}$ such that
    \begin{equation}
        |\partial^{X} s^{X} | \ge \min\left(1,\frac{2 \chi_{Q}}{\chi_Z}\right)\frac{1}{\weight_X} |g^{QX} s^{X}|
    \end{equation}
    Hence, by the Cleaning Lemma in Ref. \cite{yuan2025unified},
    \begin{align}
        d_X(C) &\ge \min\left(1,\frac{2 \chi_{Q}}{\chi_Z}\right)\frac{\chi_Z}{\weight_X} d_X(A) \\
        &= \min(\chi_Z, 2\chi_{Q}) \frac{1}{\weight_X} d_X(A)
    \end{align}
    where the extra factor $\chi_Z$ utilizes the fact that each qubit layer $R_Q\otimes R_Z^\top$ has code distance $\ge \chi_Z$. 
\end{proof}

\section*{Acknowledgements}
ACY is employed by Iceberg Quantum, and was also supported by the Laboratory for Physical Sciences at CMTC in Unversity of Maryland, College Park.
DJW is supported by the Australian Research Council Discovery Early Career Research Award (DE220100625). 
This work was initiated while NB and DJW were attending the Fault-Tolerant Quantum Technologies Workshop in Benasque, 2024. 

\bibliography{main.bbl}

\appendix
\section{Tensor Product Preserves Relative Expansion}
\label{app:TP}

In this section, we show that the tensor product preserves the isoperimetric inequality/relative expansion used in the Cleaning Lemma of Ref. \cite{yuan2025unified} (see also Ref. \cite{baspin2025fast}).
The argument essentially follows the same steps as those in Section 5.4.1 of Ref. \cite{lin2023geometrically}, with specific modifications to highlight its connection to the Cleaning Lemma.

\begin{definition}
    Let $G=E\to V$ denote a graph complex with (co)differential $\delta = \partial^{\top}$.
    Then $G$ is $c$-\textbf{(co)expanding} (at degree 0) if for any $s\in V$ (subset of vertices)
    \begin{equation}
        |\delta s| \ge \frac{2c}{|\sV|} |s| |\bar{s}|
    \end{equation}
    where $\bar{s}$ is the complement (as subsets) of $s$ in $\sV$. 
    Let $V^{\pi}$ be the space generated by a subset of vertices $\sV^{\pi} \subseteq \sV$ and $\iota:V^{\pi} \to V$ denote the inclusion map and $\pi = \iota^{\top}$ denote the projection.
    Then we say the $G$ is $c^{\pi}$\textbf{-(co)expanding} (at degree 0) \textbf{relative to} $\pi$ if for any $s\in V$ (subset of vertices)
    \begin{equation}
        |\delta s| \ge \frac{c^{\pi}}{|\sV^{\pi}|} (|\pi s||\bar{s}| + |s||\pi \bar{s}|)
    \end{equation}
    Note $\overline{\pi s} (= \pi \bar{s})$ is the complement (as subsets) in $\sV^{\pi}$.
\end{definition}
\begin{remark}
    Note that $\delta s = \delta \bar{s}$, and since $\max(|s|,|\bar{s}|)\ge |\sV|/2$, we see that coexpanding implies the expansion property that there exists $\hat{s}$ with $\delta s = \delta \hat{s}$ such that
    \begin{equation}
        |\delta s| \ge c |\hat{s}|
    \end{equation}
    While relative coexpansion implies the relative expansion property that there exists $\hat{s}$ with $\delta s = \delta \hat{s}$ such that
    \begin{equation}
        |\delta s| \ge c^{\pi}\frac{|\sV|}{|\sV^{\pi}|} |\pi\hat{s}|
    \end{equation}
\end{remark}
\begin{proof}
    The coexpansion property is clear and thus restrict our attention to relative coexpansion. Note that there exists $\hat{s}=s$ or $\bar{s}$ such that $|\pi \hat{s}|\le |\sV^{\pi}|/2$. Without loss of generality, assume $\hat{s} =s$ and thus we have
    \begin{align}
        |\pi s||\bar{s}|+|\overline{\pi s}||s| &= |\pi s| (|\sV| -|s|) + (|\sV^{\pi}|-|\pi s|)|s| \\
        &=|\pi s| |\sV| +(|\sV^{\pi}| - 2|\pi s|)|s| \\
        &\ge |\sV| |\pi s|
    \end{align}
\end{proof}
\begin{remark}
    Viewing $s\in V$ as a map $s:\sV \to \dF_2$, note that
    \begin{align}
        |\delta s| &= \sum_{e=xy\in \sE} |s(x)-s(y)| \\
        2|s||\bar{s}| &= \sum_{x,y\in \sV} |s(x)-s(y)| \\
        |\pi s||\bar{s}| + |s||\pi \bar{s}|&= \sum_{x\in V^{\pi},y\in V} |s(x)-s(y)|
    \end{align}
\end{remark}

\begin{example}[Cheeger]
    Let $G$ be a graph complex with Cheeger constant $h$. Then $G$ is $h/2$-coexpanding, and $\infty$-coexpanding relative to the trivial projection $\pi:V\to 0$.
\end{example}

\begin{example}[Repetition]
    Let $R=R(L)$ denote the repetition code on $L$ bits. Then $R$ is $2/L$-coexpanding, and $1/L$-coexpanding relative to the projection $\pi$ onto $w$ vertices, e.g., $\{|1\ket,...,|w\ket\}\subseteq R_0$
\end{example}
\begin{proof}
    If $\delta s=0$, then either $s$ or $\bar{s}$ is trivial $=0$ and thus the inequality follows for both (co)expansion and relation (co)expansion.
    Hence, consider the case where $\delta s\ne 0$.
    Note that $|s||\bar{s}| \le (|s| +|\bar{s}|)^2/4 = L^2/4$ and thus
    \begin{equation}
        |\delta s|\ge 1 \ge \frac{2}{L} \times \frac{2|s||\bar{s}|}{L}
    \end{equation}
    Hence, the repetition code is $2/L$-coexpanding.
    Note that 
    \begin{align}
        |\pi s| |\bar{s}| +|s| |\pi\bar{s}| &\le w(|\bar{s}|+|s|) \\
        &\le w L
    \end{align}
    Hence,
    \begin{equation}
        |\delta s| \ge 1 \ge \frac{1}{L}\times \frac{1}{w} ( |\pi s| |\bar{s}| +|s| |\pi\bar{s}|)
    \end{equation}
    
    Hence, the repetition code is $1/L$-coexpanding relative to $\pi$.
\end{proof}

Given graph complex $G^{A},G^{B}$ with projections $\pi^{A},\pi^{B}$ onto subsets, the tensor product $C=G^{A}\otimes G^{B}=C_2 \to C_1 \to C_0$ is a cell complex and thus also induces a graph complex $C_1 \to C_0$. The projections $\pi^{A},\pi^{B}$ also induce a projection $\pi=\pi^{A}+\pi^{B}$ onto the subspace $\pi^AV^{A}\otimes V^B \oplus V^{A} \otimes \pi^{B}V^{B}$.
Hence, we can talk about the following.

\begin{lemma}[Expansion of Tensors]
    \label{lem:expan-tensors}
    Let graph complex $G^{A}$ be $a$-coexpanding, and $a^{\pi}$-coexpanding relative to projection $\pi^{A}:V\to V^{A}$. 
    Let $G^{B}$ have similar $b$-coexpanding and $b^{\pi}$-relative coexpanding properties.
    Then $G^{A}\otimes G^{B}$ is $c$-coexpanding (at degree 0), and $c^\pi$-coexpanding relative to $\pi = \pi^{A} + \pi^{B}$ where
    \begin{align}
        c &= \min(a, b) \\
        c^{\pi} &= \min(a,b,a^{\pi},b^{\pi})
    \end{align}
    Moreover, if, say $\pi^B: V^B \to 0$ is trivial, then
    \begin{equation}
        c^{\pi} = \min(a^{\pi},b)
    \end{equation}
\end{lemma}
\begin{proof}
    Let $C=G^{A}\otimes G^{B} = C_2 \to C_1 \to C_0$ be the cell complex with induced graph complex $G=C_1 \to C_0=E\to V$.
    For notation simplicity, write $x,x'\in V^{A}$ and $y,y'\in V^{B}$ unless otherwise stated.
    Also write superscript $x^{\pi},x'^{\pi}\in \pi^{A} \sV^{A}$ and similarly for $y^{\pi},y'^{\pi} \in \pi^{B} \sV^{B}$. 
    Note that regarding $s\in V$ as a map $s:\sV \to \dF_2$, we have
    \begin{align}
        |\delta s|&= \sum_{\bra xy,x'y'\ket \in E}|s(xy)-s(x'y')| \\
        &= \sum_{y\in V^{B}} \sum_{xx'\in E^{A}} |s(xy)-s(x'y)| \nonumber\\
        &\quad\quad+\sum_{x'\in V^{A}}\sum_{yy'\in E^{B}} |s(x'y)-s(x'y')| \\
        &\ge \frac{a}{|\sV^{A}|} \sum_{y,x,x'} |s(xy)-s(x'y)| \nonumber\\
        &\quad\quad+ \frac{b}{|\sV^{B}|} \sum_{x',y,y'} |s(x'y)-s(x'y')| \\
        &\ge \frac{c}{|\sV|} \sum_{x,x',y,y'} ( |s(xy)-s(x'y)| \nonumber\\
        &\quad\quad+|s(x'y)-s(x'y')|)\\
        &\ge \frac{c}{|\sV|} \sum_{xy,x'y'} |s(xy)-s(x'y')| \\
        &= \frac{2c}{|\sV|} |s||\bar{s}|
    \end{align}
    where we used the fact that $V=V^{A} \otimes V^{B}$. Hence, it is $c$-coexpanding.

    Since $\pi s$ is the projection, we can write $s(x^{\pi}y)$ instead of $(\pi s)(x^{\pi}y)$, and similarly for $x y^{\pi}$. Note that
    \begin{align}
        &\sum_{x^{\pi}y,x'y'} |s(x^{\pi} y) - s(x'y')| \nonumber \\ 
        &\le \sum_{x^\pi y x'y'} (|s(x^{\pi} y) - s(x'y)|+ |s(x'y) - s(x'y')|) \\
        &= |\sV^{B}| \sum_{x^{\pi}x'y} |s(x^{\pi}y) - s(x'y)| \nonumber\\
        &\quad\quad +|\pi^{A} \sV^{A}| \sum_{x' y y'} |s(x'y) -s(x'y')| \\
        &\le |\sV^B| \frac{|\pi^A \sV^A|}{a^\pi}\sum_{xx'\in E^{A},y} |s(xy)-s(x'y)| \nonumber\\
        &\quad\quad +|\pi ^{A} \sV^{A}| \frac{|\sV^B|}{b} \sum_{x',yy'\in E^{B}} |s(x'y)-s(x'y')|
    \end{align}
    The case is similar for $xy^{\pi}$ so that
    \begin{align}
        &\sum_{x y^{\pi},x'y'} |s(x y^\pi) - s(x'y')|\\ 
        &\le |\sV^A| \frac{|\pi^B \sV^B|}{b^\pi}\sum_{x,yy'\in E^{B}} |s(xy)-s(xy')| \nonumber\\
        &\quad\quad +|\pi^{B} \sV^{B}| \frac{|\sV^A|}{a} \sum_{xx'\in E^{A},y'} |s(xy')-s(x'y') \\
    \end{align}
    Hence, 
    \begin{align}
        |\pi s||\bar{s}| + |\pi \bar{s}| |s| \le \frac{|\pi \sV|}{c^{\pi}} |\delta s| 
    \end{align}
\end{proof}

\begin{prop}[Square]
    \label{prop:square-expansion}
    Let $R=R(L)$ denote the repetition code on $L$ bits. Then the induced graph of $R(L_1)\otimes R(L_2)$ is $c$-coexpanding, and $c^{\pi}$-coexpanding relative to projection $\pi:R_0(L_1)\otimes R_0(L_2) \mapsto W \otimes R_0(L_2)$ where $W\subseteq R_0(L_1)$ is spanned by $w$ vertices, e.g., $\{|1\ket,...,|w\ket\}$ and
    \begin{align}
        c &= \frac{2}{\max(L_1,L_2)} \\
        c^{\pi} &=\frac{1}{L_1} \min\left(1, \frac{2L_1}{L_2}\right)
    \end{align}
    In particular, for any $s\in R_0\otimes R_0$, there exists $\hat{s}$ with $\delta s=\delta \hat{s}$, such that
    \begin{equation}
        \label{eq:square-relative-expan}
        |\delta s|\ge \min \left( 1, \frac{2 L_1}{L_2}\right) \frac{1}{w} |\pi \hat{s}|
    \end{equation}
\end{prop}

\begin{proof}
    By Lemma~\ref{lem:expan-tensors}, we see that $c,c^\pi$ follow.
    In particular, there exists $\hat{s}$ with $\delta s=\delta \hat{s}$ such that
    \begin{align}
        |\delta s| &\ge c^{\pi} \frac{|\sV|}{2|\sV^{\pi}|} |\pi\hat{s}| \\
        &\ge c^{\pi} \frac{L_1}{w} |\pi \hat{s}|\\
        &\ge \min \left( 1, \frac{2 L_1}{L_2}\right) \frac{1}{w} |\pi \hat{s}|
    \end{align}
\end{proof}


\section{Hastings's Quantum Weight Reduction}
\label{app:Hastings}

In this section, we provide a detailed discusion of Hastings's work on quantum weight reduction \cite{hastings2021quantum} using the framework in Ref. \cite{yuan2025unified}.
Roughly speaking, Hastings reduces the weight of a given CSS code by performing three individual steps -- $X$-reduction, $Z$-thickening and Hastings's coning.
The first two steps are relatively straightforward and thus collected in Propositions~\ref{prop:weight-X-reduction} and~\ref{prop:weight-Z-thicken}.
Hastings's coning procedure is more convoluted and thus will be discussed separately.
We also remark that point 5 of Lemma 8 in Ref. \cite{hastings2021quantum} has a small error -- compare with Eq.~\eqref{eq:weight-cone-corrected} of Proposition~\ref{prop:weight-cone} -- which ultimately leads to larger reduced weights, as stated in Theorem~\ref{thm:weight-corrected}, than initially claimed by Hastings.

For this section, to match the notation of Hastings \cite{hastings2021quantum}, we adopt the convention that complexes $C=C_2\to C_1 \to C_0$ denote the $Z$-checks, qubits, $X$-checks, respectively -- this is opposite to the main text -- and are equipped with basis $\sC_2,\sC_1,\sC_0$, respectively.
We shall also denote the weights of a code with diagram
\begin{equation}
C_2 \xrightleftharpoons[\qubit_Z]{\weight_Z} C_1 \xrightleftharpoons[\weight_X]{\qubit_X} C_0
\end{equation}
\subsection{Reduction and Thickening}
\begin{prop}[$X$-Reduction]
    \label{prop:weight-X-reduction}
    Let $A$ be a CSS code with weights
    \begin{equation}
        Z \xrightleftharpoons[\qubit_Z]{\weight_Z} Q \xrightleftharpoons[\weight_X]{\qubit_X} X
    \end{equation}
    and basis elements $z,q,x$ in $\sZ,\sQ,\sX$, respectively.
    Let $R(L)=E(L)\to V(L)$ denote the repetition code on $L$ vertices.
    Fix $q$ and for every $x\sim q$, assign a number $1\le i \le \qubit_X$ based on its order\footnote{For example, if $x <x'$ with respect to the ordering in $\sX$ are adjacent to $q$, we assign $i=1,2$ to $x,x'$}, which defines an injective mapping $x\mapsto (x;q)$. Similarly, define $q\mapsto (q;x) \in \{1,,...,\weight_X\}$.
    Let $C$ be obtain via the following gluing procedure
    \begin{equation}
    \label{eq:X-reduction-diagram}
    \begin{tikzpicture}[baseline]
    \matrix(a)[matrix of math nodes, nodes in empty cells, nodes={minimum size=25pt},
    row sep=1.5em, column sep=1.5em,
    text height=1.25ex, text depth=0.25ex]
    {& Z  &   \\
     & V(\qubit_X)\otimes Q  & E(\qubit_X)\otimes Q \\
     E(\weight_X) \otimes X &  V(\weight_X)\otimes X &\\};
    \path[->,font=\scriptsize]
    (a-2-2) edge node[above]{} (a-2-3);
    \path[->,font=\scriptsize]
    (a-3-1) edge node[above]{} (a-3-2);
    \path[->,font=\scriptsize]
    (a-1-2) edge[bend right=80, dashed] node[left]{$p^{XZ}_2$} (a-3-1);
    \path[->,font=\scriptsize]
    (a-1-2) edge node[right]{$g^{QZ}_2$} (a-2-2)
    (a-2-2) edge node[right]{$g^{XQ}_1$} (a-3-2);
    \end{tikzpicture}
    \end{equation}
    where 
    \begin{align}
        g^{QZ} |z\ket &= \sum_{q \sim z} \sV(\qubit_X)\otimes |q\ket, \quad \sV(\qubit_X) = \sum_{i=1}^{\qubit_X} |i\ket  \\
        g^{XQ} |i\ket \otimes q &= \sum_{x \sim q} 1\{i=(x;q)\} |(q;x)\ket\otimes |x\ket  \\
        p^{XZ} |z\ket &= \sum_{z \sim x} \sE(\weight_X)(z\wedge x;x) \otimes |x\ket
    \end{align}
    where $\sE(\weight_X)(z\wedge x;x)$ is defined as follows: note that for $x\sim z$, the common qubits $x\wedge z$ and their image in $q\mapsto (q;x)$ must be of even cardinality and thus can be paired up so that $\sE(\weight_X)(z\wedge x;x)$ denote the summation of edges between each pair.
    Then $C$, referred as the \textbf{$X$-reduced code} of $A$ satisfies $H_1(C) \cong H_1(A)$ with weights
    \begin{equation}
    C_2 \xrightleftharpoons[\weight_Z \qubit_X]{O(\weight_Z\weight_X\qubit_X)} C_1 \xrightleftharpoons[3]{3} C_0
    \end{equation}
\end{prop}
\begin{proof}
Using the framework in Definition~\ref{def:framework}, we see that
\begin{align}
    C^{Z} &= Z\to 0\to 0\\
    C^{Q} &= 0\to R(\qubit_X)^\top \otimes Q \\
    C^{X} &= 0 \to R(\weight_X) \otimes X
\end{align}
Note that the diagram (excluding the dashed line) is trivially commuting since there are no squares. Hence, to show that $C$ is a chain complex, it's sufficient to show that $g^{XQ} g^{QZ} = \partial^{X}p^{XZ}$ where we use $\partial^{Z},\partial^{Q},\partial^{X}$ to denote the differential of the levels.
Indeed, note that
\begin{align}
    g^{XQ} g^{QZ} |z\ket &= \sum_{q\sim z} \sum_{x\sim q} |(q;x)\ket\otimes |x\ket \\
    &=\sum_{x\sim z} \left(\sum_{q \in x\wedge z} |(q;x)\ket \right) \otimes |x\ket
\end{align}
And that
\begin{align}
    \partial^{X} p^{XZ} |z\ket &= \sum_{x \sim x} \left( \partial^{R(\weight_X)}\sE(\weight_X)(z\wedge x;x) \right)|x\ket \\
    &= g^{XQ} g^{QZ} |z\ket
\end{align}
where we utilized the fact that the boundary of \textit{strings} in $\sE(\weight_X)(z\wedge x;x)$ is exactly the image of $x\wedge z$ under the map $q\mapsto (q;x)$.
It's straightforward to check that the embedded column code is $A$ and thus by Theorem~\ref{thm:unified}, we see that $H_1(C)\cong H_1(A)$.
Finally, the weights of the constructed code $C$ is shown by the following diagram
\begin{equation}
    \label{eq:weight-X-reduction-diagram-2D}
    \begin{tikzpicture}[baseline]
    \matrix(a)[matrix of math nodes, nodes in empty cells, nodes={minimum size=25pt},
    row sep=1.5em, column sep=1.25em,
    text height=1.25ex, text depth=0.25ex]
    {& Z  &   \\
     & V(\qubit_X)\otimes Q  & E(\qubit_X)\otimes Q \\
     E(\weight_X) \otimes X &  V(\weight_X)\otimes X &\\};
    \path[-left to,font=\scriptsize,transform canvas={yshift=0.2ex}]
    (a-2-2) edge node[above]{2}  (a-2-3)
    (a-3-1) edge node[above]{2}  (a-3-2);
    \path[left to-,font=\scriptsize,transform canvas={yshift=-0.2ex}]
    (a-2-2) edge node[below]{2}  (a-2-3)
    (a-3-1) edge node[below]{2}  (a-3-2);
    \path[-left to,font=\scriptsize,transform canvas={xshift=0.2ex}]
    (a-1-2) edge node[right]{$\weight_Z \qubit_X$}  (a-2-2)
    (a-2-2) edge node[right]{1}  (a-3-2);
    \path[left to-,font=\scriptsize,transform canvas={xshift=-0.2ex}]
    (a-1-2) edge node[left]{$\qubit_Z$}  (a-2-2)
    (a-2-2) edge node[left]{1}  (a-3-2);
    \path[-left to,font=\scriptsize]
    (a-1-2) edge[bend right=70, dashed] node[right]{$\weight_Z\qubit_X \weight_X $ } (a-3-1);
    \path[left to-,font=\scriptsize]
    (a-1-2) edge[bend right=80, dashed] node[left]{$\weight_X \qubit_Z$} (a-3-1);
    \end{tikzpicture}
\end{equation}
Specifically, note that if $g^{ZQ}=(g^{QZ})^\top$ and similarly for $g^{QX},p^{ZX}$, then
\begin{align}
    g^{ZQ} |i\ket |q\ket &= \sum_{z\sim q} |z\ket \\
    g^{QX}|i\ket |x\ket &= \sum_{q\sim x} 1\{i=(q;x)\} |(x;q)\ket  |q\ket \\
    p^{ZX} |i^+\ket |x\ket &= \sum_{z\sim x} 1\{|i^+\ket \in \sE(\weight_X)(z\wedge x;x)\} |z\ket
\end{align}
In particular, since $q\mapsto (q;x)$ is injective, we see that $g^{QX}$ has max column weight $1$.

\end{proof}

\begin{figure}[ht]
\centering
\subfloat[\label{fig:weight-ancilla}]{%
    \centering
    \includegraphics[width=0.6\columnwidth]{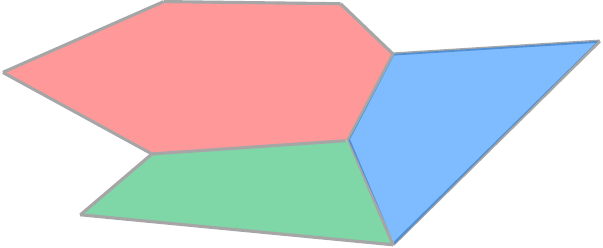}
}
\\
\subfloat[\label{fig:weight-copy}]{%
    \centering
    \includegraphics[width=0.6\columnwidth]{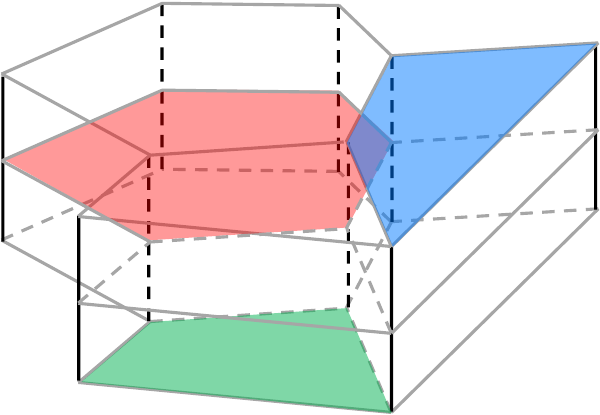}
}
\caption{Thickening. (a) depicts a cell complex $G$ corresponding to a graph with generating simple cycles $\sF$ colored. 
(b) depicts the thickened $C$ for $L=3$, where the original cycles $f\in \sF$ are mapped to different heights via $h$.
Note that in (a), there exist an edge adjacent to both red and green plaquettes, while any edge in (b) can only be adjacent to either the red or the green plaquette, and thus the weight is reduced by a clever choice of the height function.
}
\label{fig:weight-Z-thicken}
\end{figure}

\begin{prop}[$Z$-Thickening, Fig. \ref{fig:weight-Z-thicken}]
    \label{prop:weight-Z-thicken}
    Let $A$ be a CSS code with weights
    \begin{equation}
        Z \xrightleftharpoons[\qubit_Z]{\weight_Z} Q \xrightleftharpoons[\weight_X]{\qubit_X} X
    \end{equation}
    and basis elements $z,q,x$.
    Let $h:\sZ \to \{1,...,L\}$ be such that it is injective on $\{z:z\sim q\}$ for any fixed $q$ and let $C$ be obtained via the following gluing procedure
    \begin{equation}
    \label{eq:Z-thickening-diagram}
    \begin{tikzpicture}[baseline]
    \matrix(a)[matrix of math nodes, nodes in empty cells, nodes={minimum size=25pt},
    row sep=1.5em, column sep=1.5em,
    text height=1.25ex, text depth=0.25ex]
    {& Z    \\
     E(L) \otimes Q & V(L)\otimes Q   \\
     E(L) \otimes X &  V(L)\otimes X \\};
    \path[->,font=\scriptsize]
    (a-2-1) edge node[above]{} (a-2-2)
    (a-3-1) edge node[above]{} (a-3-2);
    \path[->,font=\scriptsize]
    (a-1-2) edge node[right]{$g^{QZ}_2$} (a-2-2)
    (a-2-1) edge node[right]{$g^{XQ}_2$} (a-3-1)
    (a-2-2) edge node[right]{$g^{XQ}_1$} (a-3-2);
    \end{tikzpicture}
    \end{equation}
    where $g^{XQ} =\dI \otimes \partial^{A}$ and
    \begin{equation}
        g^{QZ} |z\ket = \sum_{q \sim z} |h(z)\ket \otimes |q\ket
    \end{equation}
    Then $C$, referred as the \textbf{$Z$-thickening code} of $A$ with \textbf{height} $h$, satisfies $H_1(C)\cong H_1(A)$ and has weights
    \begin{equation}
    C_2 \xrightleftharpoons[\max(3,\weight_X)]{\max(\weight_Z,\qubit_X+2)} C_1 \xrightleftharpoons[\weight_X+2]{\max(\qubit_X,2)} C_0
    \end{equation}
\end{prop}

\begin{proof}
Using the framework in Definition~\ref{def:framework} (with $Z\lr X$ switched), we see that
\begin{align}
    C^{Z} &= Z\to 0\to 0\\
    C^{Q} &= R(L) \otimes Q \to 0 \\
    C^{X} &= 0 \to R(L) \otimes X
\end{align}
Note that
\begin{align}
    g^{XQ} g^{QZ} |z \ket &= g^{XQ} \sum_{q\sim z} |h(z)\ket\otimes |q \ket \\
    &= \sum_{x}|h(z)\ket\otimes |x\ket \left[ \sum_{q} 1\{z\sim q \sim x\}\right] \\
    &= 0
\end{align}
Where we used the fact that the common qubits adjacent to $x,z$ must be even.
It's straightforward to check that $A$ is the embedded column code and thus by Theorem~\ref{thm:unified}, $C$ is a complex with $H_1(C) \cong H_1(A)$.
Finally, the weights are given by the following diagram
\begin{equation}
\label{eq:weight-Z-thicken-diagram-2D}
\begin{tikzpicture}[baseline]
\matrix(a)[matrix of math nodes, nodes in empty cells, nodes={minimum size=25pt},
row sep=1.5em, column sep=1.5em,
text height=1.25ex, text depth=0.25ex]
{& Z    \\
 E(L) \otimes Q & V(L)\otimes Q   \\
 E(L) \otimes X &  V(L)\otimes X \\};
\path[-left to,font=\scriptsize,transform canvas={yshift=0.2ex}]
(a-2-1) edge node[above]{2}  (a-2-2)
(a-3-1) edge node[above]{2}  (a-3-2);
\path[left to-,font=\scriptsize,transform canvas={yshift=-0.2ex}]
(a-2-1) edge node[below]{2}  (a-2-2)
(a-3-1) edge node[below]{2}  (a-3-2);
\path[-left to,font=\scriptsize,transform canvas={xshift=0.2ex}]
(a-1-2) edge node[right]{$\weight_Z$}  (a-2-2)
(a-2-1) edge node[right]{$\qubit_X$}  (a-3-1)
(a-2-2) edge node[right]{$\qubit_X$}  (a-3-2);
\path[left to-,font=\scriptsize,transform canvas={xshift=-0.2ex}]
(a-1-2) edge node[left]{1}  (a-2-2)
(a-2-1) edge node[left]{$\weight_X$}  (a-3-1)
(a-2-2) edge node[left]{$\weight_X$}  (a-3-2);
\end{tikzpicture}
\end{equation}
Specifically, if $g^{ZQ}=(g^{QZ})^\top$, then
\begin{equation}
    g^{ZQ} |i\ket \otimes |q\ket = \sum_{z\sim q} 1\{i=h(z)\}|z\ket
\end{equation}
Since $h$ is injective on $\{z:z\sim q\}$, we see that $g^{ZQ}$ has max column weight $1$.
\end{proof}

\begin{remark}[Choice of Heights $h$]
    \label{rem:weight-height}
    As noted by Hastings \cite{hastings2021quantum}, one way to define the height function $h$ (and thickness $L$) is to consider the graph\footnote{This is different from that in Definition~\ref{def:induced-graphs}} consisting of vertices $\sZ$ and edges $zz'$ if $z,z'\sim q$ for some $q$.
    Let $L$ be the chromatic number and $h$ be coloring scheme. Then $L=O(\weight_Z\qubit_Z)$ and $h$ is injective on $\{z:z\sim q\}$ for any $q$.
\end{remark}

\begin{remark}[$\weight_Z$ Remains]
    \label{rem:weight-wZ-remains}
    Let $A$ be a CSS code with weights
    \begin{equation}
        Z \xrightleftharpoons[\qubit_Z]{\weight_Z} Q \xrightleftharpoons[\weight_X]{\qubit_X} X
    \end{equation}
    Apply $X$-reduction as in Proposition~\ref{prop:weight-X-reduction}, and then $Z$-thickening as in Proposition~\ref{prop:weight-Z-thicken}.
    Then the resulting complex $C$ has weights satisfying
    \begin{equation}
        C_2 \xrightleftharpoons[3]{\max(O(\weight_Z\weight_X\qubit_X),5)} C_1 \xrightleftharpoons[5]{3} C_0
    \end{equation}
    Hence, the next step is to reduce $\weight_Z(C)$.
    Note that repeating the proof (refer to Diagrams~\eqref{eq:weight-X-reduction-diagram-2D} and~\eqref{eq:weight-Z-thicken-diagram-2D}) shows that basis elements $\sC_2$ can be partitioned into $\sZ$ and its complement.
    Basis elements in the complement have small weights $\le 5$ in $C$, while those in $\sZ$ have larger support sizes $O(\weight_Z \weight_X \qubit_X)$, and thus the goal is to reduce the latter.
\end{remark}


\subsection{Hastings's Coning}

For this subsection, consider an arbitrary CSS code $A$ with weights (for simplicity, assume $\ge 3$)
\begin{equation}
    Z \xrightleftharpoons[\qubit_Z]{\weight_Z} Q \xrightleftharpoons[\weight_X]{\qubit_X} X
\end{equation}
Motivated by Remark~\ref{rem:weight-wZ-remains}, we attempt to reduce the weight of $Z$-type generators, and thus consider a fixed $z\in \sZ$.

Since the common qubits $z \wedge x$ has even cardinality, it can be grouped into pairs.
Let $\sE^{z \wedge x}$ denote the collection of pairs $|p,x;z\ket$ where $p=1,...,|x\wedge z|/2$ and $x \sim z$, and let $E^{z}$ denote the $\dF_2$ vector space generated by $\sE^{z} \coloneq \bigcup_x \sE^{z \wedge x}$.
Let $\sV^{z}$ denote the collection  of all $q\sim z$ -- denoted by $|q;z\ket$ -- and let $V^{z}$ denote that generated by $\sV^{z}$.
Define the complex $G^{z} = E^{z} \to V^{z}$ via the adjacency matrix $|q;z\ket \sim |p,x;z\ket$ if the qubit $q$ is in the pair corresponding to $|p,x;z\ket$.

Note that every 1-cell in $\sE^{z}$ has exactly two adjacent 0-cells, and thus induces a graph $\sG^{z}$ with edges $\sE^{z}$ and vertices $\sV^{z}$.
Note that if $\Omega^{z}$ denotes the collection of connected components $\omega \subseteq \sV^{z}$ of graph $\sG^{z}$, then 
a natural basis for $H^0(G^{z})$ is labeled by the connected components\footnote{The repetition code $R$ can similarly be regarded as a graph and thus one can compare with Lemma~\ref{lem:rep}.} and given by
\begin{equation}
    [\omega] = \left[\sum_{v\in \omega} v\right], \quad \omega \in \Omega^{z}
\end{equation}
where we also regard $\omega$ as a element in $V^{z}$.

Therefore, the goal is to use the gluing construction and replace $z$ by the basis of $H^0(G^z)$ corresponding to the connected components\footnote{During this process, it's possible for a connected component to induce a logical operator. However, such a logical would have weight smaller than the check weight. In Hastings's formulation, such possibilities are \textit{unreasonable} and thus initially excluded. Additional care was taken to deal with unreasonable codes \cite{hastings2021quantum}, which we do not consider in this Appendix, since the main goal is to show the correct reduced weights.}.
More specifically, a $Z$-type Pauli operator supported on qubits corresponding to $\sV^{z}$ is replaced by multiple $Z$-type Pauli operators, each one supported on qubits corresponding to a distinct component of $\sV^{z}$, and thus must have smaller supports\footnote{Technically, the support of check $z$ is only defined if the differential is given so that the number of $q \sim z$ is determined. Hence, the intuition is to later provide a differential so that the vertices in a connected component form the support of the corresponding basis element in $H^0(G^z)$}.
However, note that $G^z$ may possibly have \textit{internal logical operators}, i.e., $H^1(G^z) \cong H_1(G^{z}) \ne 0$, since there may be cycles in the graph $\sG^{z}$, and thus to remove the internal logical operators, we augment the complex $G^{z}$ by including a basis of generating cycles.
More specifically, consider the following lemma

\begin{lemma}[Decongestion \cite{hastings2021quantum,freedman2021building}]
    Let $\sG=(\sE,\sV)$ be an arbitrary graph. Then there exists a basis $\sF$ of simple cycles in $\sG$ which has total weight -- number of edges -- $O(|\sV|\log |\sV|)$, and that each edge appears in the basis at most $O(\log^2 |\sV|)$ times.
\end{lemma}

Use the previous lemma, choose the basis $\sF^{z}$ of simple cycles for $\sG^{z}$ and let $F^{z}=\spn \sF^{z}$.
Then the adjacency property defines a cell complex $C^{z} = F^{z} \to E^{z} \to V^{z}$ with no nontrivial cycles, i.e., $H_1(C^{z}) =0$, and  weights
\begin{equation}
    F^{z} \xrightleftharpoons[O(\log^2 \weight_Z)]{O(\weight_Z \log \weight_Z)} E^z \xrightleftharpoons[\qubit_X]{2} V^z
\end{equation}
Next, after applying the $Z$-thickening to the cell complex via Proposition~\ref{prop:weight-Z-thicken} we obtain the weights
\begin{equation}
    \label{eq:Z-thickened-cell-complex}
    \bm{F}^{z} \xrightleftharpoons[\qubit_X]{O(\weight_Z \log \weight_Z)} \bm{E}^z \xrightleftharpoons[\qubit_X+2]{2} V(L)\otimes V^z
\end{equation}
Hence, we what remain is to remove the large weight cycles $f$. This is achieved by cellulation as shown in Fig. \ref{fig:cellulation} and formalized by the following lemma.

\begin{lemma}[Cellulation, Fig. \ref{fig:cellulation}]
    \label{lem:cellulation}
    Let $C=F\to E\to V$ be a cell complex with weights
    \begin{equation}
        F \xrightleftharpoons[\qubit]{\weight} E \xrightleftharpoons[\Delta(C)]{2} V
    \end{equation}
    Then there exists cell complex $C'=F'\to E'\to V'$ such that $H_s(C')\cong H_s(C)$ for $s=0,1,2$ and $V'=V$ and has weights
    \begin{equation}
        F' \xrightleftharpoons[\qubit]{3} E' \xrightleftharpoons[\Delta(C) (1+2\qubit)]{2} V'
    \end{equation}
\end{lemma}

\begin{figure}[ht]
\centering
\subfloat[\label{fig:cellulation-intuition}]{%
    \centering
    \includegraphics[width=0.4\columnwidth]{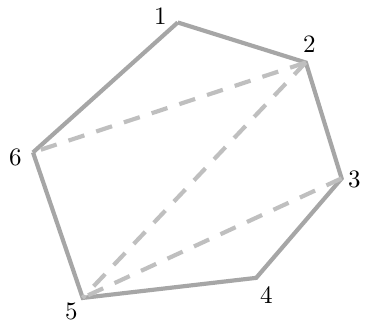}
}
\subfloat[\label{fig:cellulation-adjacency}]{%
    \centering
    \includegraphics[width=0.4\columnwidth]{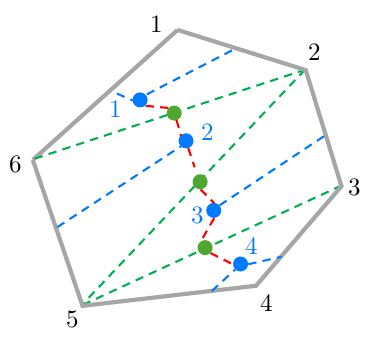}
}
\caption{Cellulation. (a) The grey lines indicates a simple cycle. The dashed lines indicate the added edges, which cellulate the simple cycle. (b) denote the adjacency relations of the repetition code $R(f)$ and the original cycle as depicted by the diagram in Eq.~\eqref{eq:cellulation-diagram}.
}
\label{fig:cellulation}
\end{figure}

\begin{proof}
    For every simple cycle/face $f\in \sF$, label the vertices via $1,...,N$ and $R(f)$ denote the repetition code on $n=\lceil (N+1)/2 \rceil$ vertices labeled by $|j;f\ket$ and edges $|j^+;f\ket$ where $j^+ = j+1/2$. Let $R(F)=E(F)\to V(F)$ denote the direct sum of $R(f),f\in \sF$.
    As depicted in Fig. \ref{fig:cellulation}, consider constructing a cell complex as follows. 
    \begin{equation}
    \label{eq:cellulation-diagram}
    \begin{tikzpicture}[baseline]
    \matrix(a)[matrix of math nodes, nodes in empty cells, nodes={minimum size=25pt},
    row sep=1.5em, column sep=1.5em,
    text height=1.25ex, text depth=0.25ex]
    {V(F) & E(F)   \\
     E & \\
     V &\\};
    \path[->,font=\scriptsize]
    (a-1-1) edge node[right]{$g^{QZ}$} (a-2-1)
    (a-2-1) edge node[right]{$g^{XQ}$} (a-3-1)
    (a-1-1) edge (a-1-2) 
    (a-1-2) edge[bend left,dashed] node[right]{$p^{XZ}$} (a-3-1);
    \end{tikzpicture}
    \end{equation}
    Here, $g^{XQ}=\partial^{G}$, and $g^{QZ}$ is defined to map $|j;f\ket \in V(f)$ to edge(s) in $f$ with endpoints (recall $n=\lceil (N+1)/2\rceil$)
    \begin{itemize}
        \item $N,1$ and $1,2$ for $j=1$
        \item $\ell,\ell+1$ if $j=2\ell-1\ne 1,n$
        \item $N-\ell+1,N-\ell$ if $j=2\ell \ne n$ 
        \item $n-1,n$ and $n,n+1$ if $j=n$
    \end{itemize}
    And $p^{XZ}$ maps edges $|j^+;f\ket$ to vertices in $f$ labeled by $\lceil j/2\rceil +1$ and $n-\lfloor j/2\rfloor$.
    It's straightforward to check that $g^{QZ},g^{XQ},p^{XZ}$ satisfy the compatibility conditions. The weights are then given by
    \begin{equation}
    \label{eq:Hastings-coning-diagram}
    \begin{tikzpicture}[baseline]
    \matrix(a)[matrix of math nodes, nodes in empty cells, nodes={minimum size=25pt},
    row sep=1.5em, column sep=5em,
    text height=1.25ex, text depth=0.25ex]
    {V(F)  & E(F)   \\
     E & \\
     V &\\};
    \path[-left to,font=\scriptsize,transform canvas={yshift=0.2ex}]
    (a-1-1) edge node[above]{2}  (a-1-2);
    \path[left to-,font=\scriptsize,transform canvas={yshift=-0.2ex}]
    (a-1-1) edge node[below]{2}  (a-1-2);
    \path[-left to,font=\scriptsize,transform canvas={xshift=0.2ex}]
    (a-1-1) edge node[right]{2}  (a-2-1)
    (a-2-1) edge node[right]{2}  (a-3-1);
    \path[left to-,font=\scriptsize,transform canvas={xshift=-0.2ex}]
    (a-1-1) edge node[left]{$\qubit$}  (a-2-1)
    (a-2-1) edge node[left]{$\Delta(C)$}  (a-3-1);
    \path[-left to,font=\scriptsize]
    (a-1-2) edge[bend left=50, dashed] node[right]{2} (a-3-1);
    \path[left to-, font=\scriptsize]
    (a-1-2) edge[bend left=40, dashed] node[left]{$2\qubit \Delta(C)$} (a-3-1);
    \end{tikzpicture}
    \end{equation}
    In particular, the weight indicated by the dashed line between $V\to E(F)$ is due to the upper bound $\qubit \Delta(G)$ on the number of faces $f$ adjacent to a vertex $v$, and that for each adjacent cycle $f$, a vertex may be adjacent to two new added edges as depicted in Fig. \ref{fig:cellulation}.
\end{proof}
Triangulate/cellulate each face in Eq.~\eqref{eq:Z-thickened-cell-complex} as described in Lemma~\ref{lem:cellulation} so that the resulting complex has weights
\begin{equation}
    \label{eq:associated-complex}
    \hat{\bm{F}}^{z} \xrightleftharpoons[\qubit_X]{3} \hat{\bm{E}}^z \xrightleftharpoons[(\qubit_X+2)(1+2\qubit_X)]{2} V(L)\otimes V^z
\end{equation}
where the faces and edges are modified, but the vertices remain the same.
The essential result is that the associated cell complex $\hat{\bm{C}}^{z}$ has weights with upper bound only depending on $\qubit_X$, has no \textit{internal logical operators}, i.e., $H_1(\hat{\bm{C}}^{z}) =0$, and has the following natural basis\footnote{Intuitively, this is seen from Fig. \ref{fig:weight-Z-thicken}, which shows that thickening only changes the connected components via $\sV(L)$, while cellulation does not change the connectivity.} for $H^0(\hat{\bm{C}}^z)$  (where we slightly abuse notation)
\begin{equation}
    \label{eq:weight-components}
    [\omega] \equiv \left[\sV(L)\right] \otimes \left[\sum_{v\in \omega} v\right],\quad \omega\in \Omega^{z}
\end{equation}

\begin{prop}[Hastings's Coning]
    \label{prop:weight-cone}
    Let $A$ be a CSS code with weights (assumed to be $\ge 3$ for simplicity)
    \begin{equation}
        Z \xrightleftharpoons[\qubit_Z]{\weight_Z} Q \xrightleftharpoons[\weight_X]{\qubit_X} X
    \end{equation}
    Let $\hat{\bm{C}}^{Z}\equiv \bigoplus \hat{\bm{C}}^{z}$ be the direct sum of the associated $\hat{\bm{C}}^{z}$ defined in Eq.~\eqref{eq:associated-complex}.
    Then consider $C$ via the following gluing procedure
    \begin{equation}
    \label{eq:Hastings-coning-diagram}
    \begin{tikzpicture}[baseline]
    \matrix(a)[matrix of math nodes, nodes in empty cells, nodes={minimum size=25pt},
    row sep=1.5em, column sep=1.5em,
    text height=1.25ex, text depth=0.25ex]
    {\hat{\bm{V}}^Z  & \hat{\bm{E}}^Z & \hat{\bm{F}}^Z   \\
     Q && \\
     X &&\\};
    \path[->,font=\scriptsize]
    (a-1-1) edge node[right]{$g^{QZ}$} (a-2-1)
    (a-2-1) edge node[right]{$g^{XQ}$} (a-3-1)
    (a-1-1) edge (a-1-2) 
    (a-1-2) edge (a-1-3)
    (a-1-2) edge[bend left,dashed] node[right]{$p^{XZ}$} (a-3-1);
    \end{tikzpicture}
    \end{equation}
    where $g^{XQ} = \partial^{A}$ and
    \begin{align}
        g^{QZ} |l\ket \otimes |q;z\ket &= 1\{l=1\} |q \ket \\
        p^{XZ} |l\ket \otimes |p,x;z\ket &= 1\{l=1\} | x \ket
    \end{align}
    where $p^{XZ}$ acts nontrivially only on $V(L)\otimes E^{Z}$ (refer to Diagram~\eqref{eq:Z-thickening-diagram} and Eq.~\eqref{eq:Z-thickened-cell-complex} ).
    Then $C$ is a complex such that $H_1(C) \cong H_1(C^{\bg})$ where $C^{\bg}$ is the embedded column complex, and has weights
    \begin{equation}
        \label{eq:weight-cone-corrected}
        C_2 \xrightleftharpoons[\qubit_Z]{(\qubit_X +2)(1+2\qubit_X)+1} C_1 \xrightleftharpoons[\weight_X + \weight_X^2 \qubit_Z/2]{\qubit_X+1} C_0
    \end{equation}
\end{prop}
\begin{proof}
Using the framework in Definition~\ref{def:framework}, we see that
\begin{align}
    C^{Z} &= (\hat{\bm{C}}^Z)^{\top}\\
    C^{Q} &= 0\to Q \to 0 \\
    C^{X} &= 0 \to 0 \to X
\end{align}
It's then straightforward to check that $p^{XZ}\partial^{Z} = g^{XQ} g^{QZ}$ and that $H_1(C) \cong H_1(C^{\bg})$.
Finally, the weights follow from the following diagram
\begin{equation}
    \begin{tikzpicture}[baseline]
    \matrix(a)[matrix of math nodes, nodes in empty cells, nodes={minimum size=25pt},
    row sep=1.5em, column sep=5em,
    text height=1.25ex, text depth=0.25ex]
    {V  & E & F   \\
     Q && \\
     X &&\\};
    \path[-left to,font=\scriptsize,transform canvas={yshift=0.2ex}]
    (a-1-1) edge node[above]{$(\qubit_X+2)(1+2\qubit_X)$}  (a-1-2)
    (a-1-2) edge node[above]{$\qubit_X$}  (a-1-3);
    \path[left to-,font=\scriptsize,transform canvas={yshift=-0.2ex}]
    (a-1-1) edge node[below]{2}  (a-1-2)
    (a-1-2) edge node[below]{3}  (a-1-3);
    \path[-left to,font=\scriptsize,transform canvas={xshift=0.2ex}]
    (a-1-1) edge node[right]{1}  (a-2-1)
    (a-2-1) edge node[right]{$\qubit_X$}  (a-3-1);
    \path[left to-,font=\scriptsize,transform canvas={xshift=-0.2ex}]
    (a-1-1) edge node[left]{$\qubit_Z$}  (a-2-1)
    (a-2-1) edge node[left]{$\weight_X$}  (a-3-1);
    \path[-left to,font=\scriptsize]
    (a-1-2) edge[bend left=50, dashed] node[right]{1} (a-3-1);
    \path[left to-, red, font=\scriptsize]
    (a-1-2) edge[bend left=40, dashed] node[left]{$\frac{1}{2} \weight_X^2 \qubit_Z$} (a-3-1);
    \end{tikzpicture}
\end{equation}
In particular, the red dashed line denotes the fact that for every $x\in \sX$, there at most $\weight_X \qubit_Z$ many $z\sim x$, each of which has $|x\wedge z|/2 \le \min(\weight_X,\weight_Z)/2$ pairs. Since we expect $\weight_Z$ to be the limiting factor as in Remark~\ref{rem:weight-wZ-remains}, we write $|x\wedge z| \le \weight_X/2$.
This contribution was not considered in point 5 of Lemma 8 in Hastings \cite{hastings2021quantum}.
\end{proof}

Note the unlike the previous propositions, the embedded column code of $C$ of Proposition~\ref{prop:weight-cone} is not exactly equal to $A$.
However, if we compare the embedded code $C^{\bg}$ in the previous proposition with the code $A$, we see that the only difference is that basis elements $z\in \sZ$ are replaced by those of $H^0(\hat{\bm{C}}^z) \cong H^0(C^{z})$ as specified by Eq.~\eqref{eq:weight-components}.
Because basis elements of $H^0(C^z)$ correspond to connected components of generators $z\in \sZ$, it's clear that, after suitable isomorphisms, the basis of $H^0(C^z)$ generates any $z$, and thus what remains is whether every basis element of $H^0(C^Z)$ can be generated by $\sZ$.
Therefore, Hastings considered\footnote{Hastings developed a further method to circumvent this issue and remove the requirement of $A$ being reasonable, but since this is not the main goal of this appendix, we do not go into further detail here.} CSS codes that are \textbf{$Z$-type reasonable}, i.e., there are no $Z$-type logical operators whose support is contained in the support of some $Z$-type generator, and thus if $A$ is reasonable, then the basis elements of $H^0(C^{z})$ cannot correspond to nontrivial logical $Z$-type operators.
This discussion is summarized more precisely as follows.

\begin{prop}[Reasonable Codes]
    Let $A =Z\to Q\to X$ be a CSS code which is \textbf{$Z$-type reasonable}, i.e., if $\ell^{A}$ is a nontrivial $Z$-type logical operator of $A$ regarded as a subset of $\sQ$, then $\ell^{A}$ is not a subset of $\partial^{A} z$ for any $z\in \sZ$.
    Let $C^{\bg}$ be the embedded code in Proposition~\ref{prop:weight-cone}. 
    Then $\im \partial^{\bg}_{2} = \im \partial^{A}_{2}$. 
    In particular,
    \begin{equation}
        H_1(C^{\bg}) = H_1(A)
    \end{equation}
\end{prop}

\begin{proof}
    Note that for component $\omega\in \Omega^{z}$
    \begin{align}
        \partial^{\bg} [\omega] &= \left[g^{QZ} \sV(L)\otimes \sum_{q}1\{q\in \omega\}   |q;z\ket\right]\\
        &= \sum_{q}1\{q\in \omega\} |q \ket \equiv \omega
    \end{align}
    where $\omega$ is also regarded as an element in $Q$ via the natural identification.
    It's clear that if we sum over all connected components $\omega \in \Omega^{z}$, the right-hand-side will be $\partial^{A} z$ and thus $\im \partial^{A}_{2} \subseteq \im \partial_{2}^{\bg}$.    
    If $\im \partial^{A}_{2} $ is a proper subset of $ \im \partial_{2}^{\bg}$, there must exist some $\omega\in \Omega^{z}$ for some $z$ such that $\omega \notin \im \partial^{A}_{2}$, and thus $\omega$ is a nontrivial $Z$-type logical operator (with respect to $A$) such that, regarded as a subset, is contained in $\partial^{A} z$.
    Hence, we reach a contradiction.
\end{proof}

\begin{remark}[Partially $Z$-type Reasonable Codes]
    \label{rem:partially-reasonable-codes}
    One can also talk about codes $A=Z\to Q\to X$ which are $Z$-type reasonable with respect to a subset $\sZ'\subseteq \sZ$. In this case, the construction in Proposition~\ref{prop:weight-cone} should be modified as follows.
    Let $\hat{\bm{C}}^{Z'}$ be the direct sum of associated $\hat{\bm{C}}^z$ (as defined in Eq.~\eqref{eq:associated-complex}) over all $z\in \sZ'$, and let $Z'' =\spn \sZ''$ where $\sZ'' = \sZ \backslash \sZ$. Then consider the cell complex constructred via
    \begin{equation}
    \label{eq:Hastings-coning-diagram}
    \begin{tikzpicture}[baseline]
    \matrix(a)[matrix of math nodes, nodes in empty cells, nodes={minimum size=25pt},
    row sep=1.5em, column sep=1.5em,
    text height=1.25ex, text depth=0.25ex]
    {Z''\oplus \hat{\bm{V}}^{Z'}  & \hat{\bm{E}}^{Z'} & \hat{\bm{F}}^{Z'}   \\
     Q && \\
     X &&\\};
    \path[->,font=\scriptsize]
    (a-1-1) edge node[right]{} (a-2-1)
    (a-2-1) edge node[right]{} (a-3-1)
    (a-1-1) edge (a-1-2) 
    (a-1-2) edge (a-1-3)
    (a-1-2) edge[bend left,dashed] node[right]{} (a-3-1);
    \end{tikzpicture}
    \end{equation}
    In this case, $H_1(C) \cong H_1(A)$ whereas the weight diagram is given by
    \begin{equation}
        C_2 \xrightleftharpoons[\qubit_Z]{\max((\qubit_X +2)(1+2\qubit_X)+1,\weight_{Z''})} C_1 \xrightleftharpoons[\weight_X + \weight_X^2 \qubit_Z/2]{\qubit_X+1} C_0
    \end{equation}
    where $\weight_{Z''}$ is the max weight of checks in $\sZ''$ with respect to $A$.
\end{remark}


\subsection{Final Result}

Let us now attempt to combine Propositions~\ref{prop:weight-X-reduction},~\ref{prop:weight-Z-thicken} and~\ref{prop:weight-cone}.
First note that 
\begin{lemma}
    Let $A=Z\to Q\to X$ be a $Z$-type reasonable CSS code.
    \begin{itemize}
        \item If $C$ is the $X$-reduction of $A$, then $C$ is $Z$-type reasonable.
        \item If $C$ is the $Z$-thickening of $A$, then $C$ is $Z$-type reasonable with respect to $\sZ \subseteq \sC_2$.
    \end{itemize}
\end{lemma}
\begin{proof}
    Let $C$ be the $X$-reduction of $A$ so that $\sC_2 = \sZ$ via Diagram~\eqref{eq:weight-X-reduction-diagram-2D}.
    Suppose $C$ is not $Z$-type reasonable, so that there exists nontrivial logical operator $\ell^{C} \in C_{1}$ such that, regarded as a subset, $\ell^{C} \subseteq \partial^{C} z$ for some $z \in \sZ$.
    Note that $\ell^{C} = \ell^{Q} \oplus \ell^{X}$ where $\ell^{Q}\in V(\qubit_X)\otimes Q$ and $\ell^{X} \in E(\weight_X) \otimes X$.
    Since $\partial^C \ell^{C} =0$, we see that $\delta^{R(\qubit_X)} \ell^{Q}=0$, and thus $\ell^{Q}$ must be of the form $\sV(\qubit_X)\otimes \ell^{A}$ for some $\ell^{A} \in Q$.
    Also note that
    \begin{equation}
        g^{XQ} \ell^{Q} = \sum_{x} \underbrace{\left(\sum_{\ell^{A} \ni q \sim x} |(q;x)\ket\right)}\otimes |x \ket
    \end{equation}
    Since $\partial^{C}\ell^{C} =0$, we see that the underbraced term must be of even cardinality for all $x\in \ell^{A}$. Hence, $\partial^{A} \ell^{A} =0$ and thus $\ell^{A}$ is a $Z$-type logical operator of $A$.
    Since $\ell^{C} \subseteq \partial^{C} z$ for some $z$, we see that $\ell^{A} \subseteq \partial^{A}z$. Since $\ell^{C}$ is a nontrivial logical of $C$, then so is $\ell^{A}$. This contradicts the fact that $A$ is $Z$-type reasonable.
    
    
    Now let $C$ be the $Z$-thickening of $A$ so that $\sZ$ is naturally a subset of $\sC_2$ via Diagram~\ref{eq:weight-Z-thicken-diagram-2D}.
    Suppose $C$ is not $Z$-type reasonable with respect to $\sZ$, so that there exists nontrivial logical operator $\ell^{C}\in C_{1}$ such that, regarded as a subset, $\ell \subseteq \partial^{C} z$ for some $z\in \sZ$.
    Then there exists $\ell^{A} \subseteq \partial^{A} z$ such that
    \begin{equation}
        \ell^{C} = |h(z)\ket \otimes \ell^{A}
    \end{equation}
    Since $\partial^{C} \ell^{C} =0$, we see that $\partial^{A} \ell^{A} =0$ and thus $\ell^{A}$ is a $Z$-type logical operator. SInce $\ell^{C}$ is nontrivial, so is $\ell^{A}$ and thus we reach a contradiction since $A$ is $Z$-type reasonable.
\end{proof}
\begin{theorem}[Hastings's Weight Reduction -- Corrected]
    \label{thm:weight-corrected}
    Let $A$ be a $Z$-type reasonable CSS code with weights 
    \begin{equation}
        Z \xrightleftharpoons[\qubit_Z]{\weight_Z} Q \xrightleftharpoons[\weight_X]{\qubit_X} X
    \end{equation}
    Let $C$ denote the chain complex obtained by the sequential construction of $X$-reduction, $Z$-thickening and Hastings's coning.
    Then $H_1(C) \cong H_1(A)$ and $C$ has weights satisfying
    \begin{equation}
        Z \xrightleftharpoons[3]{36} Q \xrightleftharpoons[42]{4} X
    \end{equation}
\end{theorem}
\begin{proof}
    By Remark~\ref{rem:weight-wZ-remains}, start from the chain complex $D$ after $X$-reduction and $Z$-thickening with choice of height function described in Remark~\ref{rem:weight-height}.
    As suggested by Remark~\ref{rem:weight-wZ-remains}, partition $\sD_2$ into $\sZ$ and its complement.
    Since $A$ is $Z$-type reasonable, by the previous lemma, we see that $D$ is $Z$-type reasonable with respect to $\sZ \subseteq \sD_2$.
    Let $C$ be the Hastings's cone constructed from replacing each $z\in \sZ$ with its associated cell complex $\hat{\bm{C}}^{z}$. 
    Then the statement follows from Proposition~\ref{prop:weight-cone} and Remark~\ref{rem:partially-reasonable-codes}.
\end{proof}
    

\end{document}